\documentclass{article}
\usepackage{fullpage}
\usepackage[english]{babel}

\usepackage{microtype}
\usepackage{graphicx}
\usepackage{subcaption}
\usepackage{tikz}
\usepackage{booktabs}
\usepackage{hyperref}

\usepackage{amsmath}
\usepackage{amssymb}
\usepackage{mathtools}
\usepackage{amsthm}

\usepackage[capitalize,noabbrev]{cleveref}

\theoremstyle{plain}
\newtheorem{theorem}{Theorem}[section]
\newtheorem{example}{Example}

\title{Binding threshold units with artificial oscillatory neurons}

\author{Vladimir Fanaskov\\
AIRI, Skoltech\\
\texttt{fanaskov.vladimir@gmail.com} \and Ivan Oseledets \\
AIRI, Skoltech
}
\date{}

\begin{document}
\maketitle
\begin{abstract}
Artificial Kuramoto oscillatory neurons were recently introduced as an alternative to threshold units. Empirical evidence suggests that oscillatory units outperform threshold units in several tasks including unsupervised object discovery and certain reasoning problems. The proposed coupling mechanism for these oscillatory neurons is heterogeneous, combining a generalized Kuramoto equation with standard coupling methods used for threshold units. In this research note, we present a theoretical framework that clearly distinguishes oscillatory neurons from threshold units and establishes a coupling mechanism between them. We argue that, from a biological standpoint, oscillatory and threshold units realise distinct aspects of neural coding: roughly, threshold units model intensity of neuron firing, while oscillatory units facilitate information exchange by frequency modulation. To derive interaction between these two types of units, we constrain their dynamics by focusing on dynamical systems that admit Lyapunov functions. For threshold units, this leads to Hopfield associative memory model, and for oscillatory units it yields a specific form of generalized Kuramoto model. The resulting dynamical systems can be naturally coupled to form a Hopfield-Kuramoto associative memory model, which also admits a Lyapunov function. Various forms of coupling are possible. Notably, oscillatory neurons can be employed to implement a low-rank correction to the weight matrix of a Hopfield network. This correction can be viewed either as a form of Hebbian learning or as a popular LoRA method used for fine-tuning of large language models. We demonstrate the practical realization of this particular coupling through illustrative toy experiments.
\end{abstract}

\section{Introduction}
Artificial neural network consists of simplified abstract artificial neurons that follow principles of McCulloch–Pitts model \cite{mcculloch1943logical}. The elementary constituents, artificial neurons or threshold units\footnote{Formally, we consider ``second generation'' neurons with continuous activation functions \cite{maass1997networks}.}, linearly integrate incoming information and apply simple nonlinearity before passing it to the next neuron. In approximation theory such representations are known as superpositions and demonstrate remarkable properties not found in linear approximation schemes \cite{cybenko1989approximation}, \cite{barron1993universal}, \cite{yarotsky2017error}. The threshold unit is not a realistic model of a biological neuron but is well-suited to capture time-averaged neuron interactions in a simplified abstract way \cite{maass1997networks}. Organized into layers, these artificial neurons form hierarchical multilayer networks capable of modeling both feedforward and feedback connections, structures also crucial in biological systems \cite{lindsay2021convolutional}. Despite their simplicity, networks of these units have achieved remarkable state-of-the-art performance across a wide range of cognitive tasks \cite{schmidhuber2015deep}, \cite{lecun2015deep}.

More recently, a novel oscillatory artificial neuron was proposed \cite{miyato2024artificial}, aiming to incorporate the crucial role of oscillatory dynamics observed in biological brain functions \cite{buzsaki2004neuronal}, \cite{duzel2010brain}, \cite{headley2017common}, \cite{van2025processes}, \cite{keller2024spacetime}. Mathematically, a key property of these oscillatory neurons is their capacity for synchronization, often modeled by the Kuramoto equation \cite{kuramoto1984chemical}, \cite{dorfler2014synchronization}. This equation describes how interconnected oscillators tend to synchronize their rhythms by exchanging information based on the oscillation phase differences. This synchronization mechanism allows for information propagation across the network and the formation of dynamically coherent groups. Oscillatory neurons, and related models have already shown promising results in tasks such as object discovery, segmentation, and reasoning \cite{miyato2024artificial}, \cite{lowe2024rotating}, \cite{liboni2025image}.

Dynamics of oscillatory artificial neurons is fairly constrained by the form of generalized Kuramoto equation. To enhance the expressiveness of these units, the work in \cite{miyato2024artificial} introduced several simplifying assumptions. As a result, proposed architecture is a complicated mixture of Kuramoto update and standard deep learning layers which is performant, but lacking theoretical justification and coherence.

In this research note, we present a more theoretically grounded approach to integrate threshold units with oscillatory neurons. Our central argument is that threshold units and oscillatory units model distinct aspects of biological neuronal interaction, compel a clear separation in their dynamical updates. Specifically, threshold units model interactions based on averaged neuronal activity, akin to the Hopfield associative memory model \cite{hopfield1984neurons}, while oscillatory units capture interactions through frequency modulation, following the generalized Kuramoto model \cite{lohe2009non}. To bridge these two types of units, we demonstrate that under specific technical conditions, a generalized Kuramoto model possesses a global energy function, which can be naturally extended to serve as a Lyapunov function for our proposed Hopfield-Kuramoto associative memory model that contain coupling between oscillatory and threshold neurons. We offer several interpretations for the resulting interaction terms, notably they can be understood as a form of ``fast weights'' \cite{ba2016using}, a time-dependent low-rank correction (LoRA) \cite{hu2022lora}, or a Hebbian learning mechanism for the weight matrix of threshold and oscillatory units \cite{hopfield1982neural}, \cite[Definition 2]{albanese2024hebbian}.

The concise list of contribution is as follows:
\begin{enumerate}
	\item We itroduce coupling between artificial oscillatory and threshold units.
	\item To theoretically justify this coupling, a novel Hopfield-Kuramoto associative memory model with a provable global energy function is proposed.
	\item Our work develops a general framework for constructing deep neural networks that integrate both oscillatory and threshold units.
	\item Several interpretations of the derived coupling mechanisms are offered, including connections to fast weights, LoRA, and Hebbian learning.
	\item Simple numerical experiments validate the effectiveness of a low-rank coupling mechanism.
\end{enumerate}

\section*{Notation}
Vectors, vector functions, matrices are in bold font and their components are in regular font. For example, $\boldsymbol{g}(\boldsymbol{x})$ is a vector function of vector argument $\boldsymbol{x}$, $g(\boldsymbol{x})_3$ is its third component; $\boldsymbol{W}_{ij}$ is a block matrix; $\boldsymbol{I}_{N}\in\mathbb{R}^{N\times N}$ is identity matrix, subscript is omitted when dimension is evident from the context; $\boldsymbol{1}$ is identity vector, i.e, $\left(\boldsymbol{1}\right)_i = 1$ for all components.

We use $\dot{f}$ to indicate time derivatives. In all initial-value problems initial conditions are assumed to be given. We omit them when appropriate.

Matrix operations that we use include Hadamard product $\left(\boldsymbol{A}\odot\boldsymbol{B}\right)_{ij} = A_{ij}B_{ij}$, tensor product
\begin{equation*}
	\begin{pmatrix}
		b_{1} & b_{2} \\
	\end{pmatrix}
	\otimes
	\begin{pmatrix}
		a_{11} & a_{12} \\
		a_{21} & a_{22}
	\end{pmatrix} =
	\begin{pmatrix}
		b_{1}a_{11} & b_{1}a_{12} & b_{2}a_{11} & b_{2}a_{12} \\
		b_{1}a_{21} & b_{1}a_{22} & b_{2}a_{21} & b_{2}a_{22}
	\end{pmatrix},
\end{equation*}
Gram matrix $\left(\boldsymbol{\mathcal{G}}(\boldsymbol{a}, \boldsymbol{b})\right)_{ij} = \boldsymbol{a}_i^{\top} \boldsymbol{b}_j$ for two sets of vectors $\boldsymbol{a}_i,i\in1,\dots, N$, $\boldsymbol{b}_j,j\in1,\dots, M$, $\boldsymbol{\mathcal{G}}(\boldsymbol{a}, \boldsymbol{b}) \in \mathbb{R}^{N\times M}$, and trace $\left\|\boldsymbol{A}\right\|_{F} = \sqrt{\sum_{ij} \left(A_{ij}\right)^2}$.

Gradient of scalar function is $\left(\frac{\partial L(\boldsymbol{x})}{\partial \boldsymbol{x}}\right)_i = \frac{\partial L(\boldsymbol{x})}{\partial x_i}$ and Hessian of scalar function reads  $\left(\frac{\partial^2 L(\boldsymbol{x})}{\partial \boldsymbol{x}^2}\right)_{ij} = \frac{\partial^2 L(\boldsymbol{x})}{\partial x_i \partial x_j}$; $\delta_{ij}$ refers to components of identity matrix, $\exp(\boldsymbol{A}) = \sum_{k=0}^{\infty}\boldsymbol{A}^{k}/k!$ is a standard matrix exponent.

\section{Threshold units and artificial oscillatory neurons}
\label{section:Threshold_and_oscillatory}
We start by describing models of threshold units and oscillatory neurons along with their biological interpretations.

The intricate and incompletely understood nature of biological neuron interactions has led to a diverse array of neuron models, each varying in its level of biological plausibility. One particularly convenient biologically realistic model was proposed by Izhikevich in \cite{izhikevich2003simple}. Employing just two ordinary differential equations and an event-based reset condition, this model effectively captures a wide range of biologically relevant neuronal behaviors.. The membrane potential under constant injected DC current predicted by the Izhikevich model is available on the central panel of Figure~\ref{fig:three_types_of_neurons}. The dynamics is clearly non-smooth owing to the presence of a reset mechanism in the model. This inherent non-smoothness is a key challenge in training spiking neural networks\footnote{While specialized training methods exist \cite{eshraghian2023training}, classical deep learning techniques often exhibit superior performance. This motivates the exploration of simplified, smoother models.} A common approach to circumvent this training challenge is to employ simplified models that abstract away certain biological complexities while offering more desirable numerical properties, such as smoother dynamics. The threshold unit is a prime example of such a model.

\begin{figure*}
	\begin{tikzpicture}
		
		\node at (0, 0) (empty_node)  {};
		
		\begin{scope}[shift={(2.5,0.25)}]
			\draw[thick,-] (-1, -0.83) -- (1, -0.83);
			\draw[thick,-] (-0.53, -1) -- (-0.53, 1);
			\draw[scale=0.15, domain=-6.5:6.5, smooth, variable=\x, black, very thick] plot ({\x}, {10 / (1 + exp(-\x)) - 5});
			\filldraw[black] (-0.2, -0.45) circle (2pt);
			\draw[very thick,-] (-0.2, -0.83) -- (-0.2, -0.45);
			\draw[] (-1.35, 0.85) node[right] {$\sigma(x)$};
			\draw[] (0.7, -1.05) node[right] {$x$};
			\filldraw[black] (0.3, 0.57) circle (2pt);
			\draw[very thick,-] (0.3, -0.83) -- (0.3, 0.56);
		\end{scope}
		
		\draw [black, thick, ->] plot [smooth, tension=2] coordinates {(7.5, 0.8) (4.6, 2.0) (2.2, 0.1)};
		\draw [black, thick] plot [smooth, tension=1] coordinates {(5.5, 0.2) (4.3, 0.5) (3.8, 1.5) (3.24, 1.661)};
		\draw [black, thick, ->] plot [smooth, tension=1] coordinates {(5.5, -1.5) (4.3, -1.1) (3.9, 0.2) (3.0, 0.75)}; 
		
		\begin{scope}[shift={(14.2,0.11)}]
			\draw[thick,-] (-1, 0) -- (1, 0);
			\draw[thick,-] (0, -1) -- (0, 1) ;
			\draw[thick](0,0) circle (0.7);
			\draw[very thick] (0, 0) -- (0.4, 0.57);
			\filldraw[black] (0.4, 0.57) circle (2pt);
			\draw[very thick] (0, 0) -- (-0.6, 0.36);
			\filldraw[black] (-0.6, 0.36) circle (2pt);
			\draw[] (0.69, 0.22) node[right] {$\phi$};
		\end{scope}
		
		\draw [black, thick, ->] plot [smooth, tension=2] coordinates {(8.5, 0.5) (12.0, 2.0) (14.7, 0.9)};
		\draw [black, thick, ->] plot [smooth, tension=2] coordinates {(8.5, -1.0) (9.5, -0.1) (13.4, 0.5)};
		
		\node at (14.2, -1.5) (interaction)  {$w_{12}\sin(\phi_{1} - \phi_{2})$};
		
		\node at (2.6, -1.5) (interaction)  {$w_{12}\sigma(x_2)$};
		
		\node at (8.5, 0)  (image) {\includegraphics[width=0.35\textwidth]{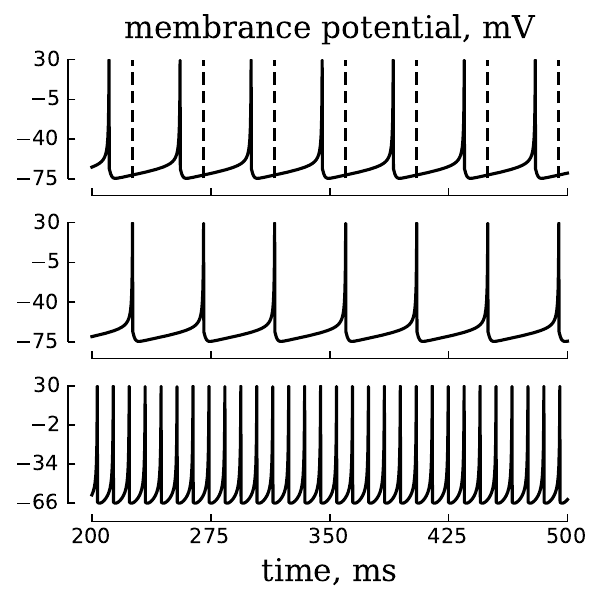}};
	\end{tikzpicture}
	\caption{\textbf{In the middle:} membrane potential of three Izhikevich neurons under constant injected dc-current. Two first neurons fire at the same frequency but have different phases. The third neuron spikes at a higher frequency. \textbf{On the left:} threshold unit with smooth activation function $\sigma(x)$; interaction term of additive model. Threshold unit is a simplified description of neuron's interaction that only models time-averaged intensity of spikes $x$. Phase shift is ignored by the model. \textbf{On the right:} artificial oscillatory neuron with $D=1$; interaction term of Kuramoto model. Artificial oscillatory neuron is a simplified model that describes a neuron as a harmonic oscillator with fixed natural frequency. Interaction of artificial oscillatory neurons depends only on the phase difference. According to \cite{izhikevich1999weakly} oscillatory neurons with different non-resonant frequencies do not interact, and in this sense average intensity of spikes is ignored by the model. See Section~\ref{section:Threshold_and_oscillatory} for discussion.}
	\label{fig:three_types_of_neurons}
\end{figure*}
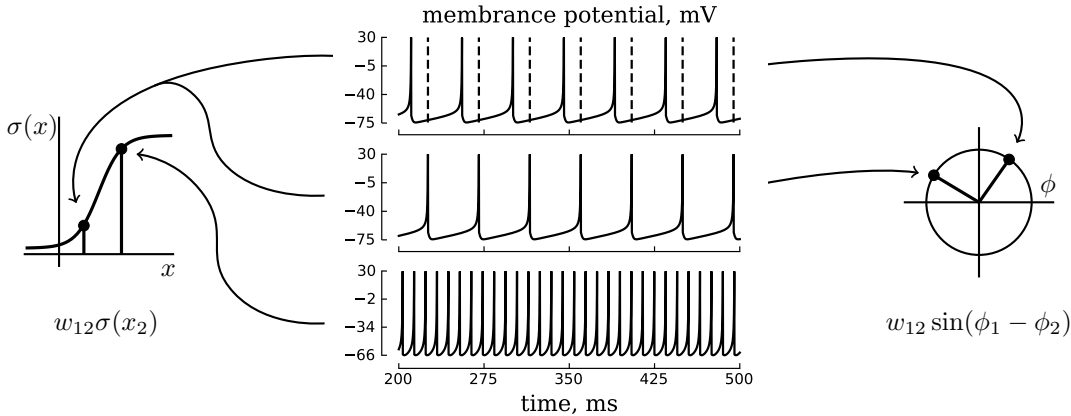

The McCulloch-Pitts artificial neuron, or threshold unit, represents the time-averaged activity of biological neurons, such as their average firing frequency, through a scalar state variable $x$. The interaction with other artificial neurons is 
\begin{equation}
	\label{eq:threshold_unit_dynamics}
	\dot{x}_{i} = \sum_{j}W_{ij} \sigma\left(x_{j}\right) - x_i,
\end{equation}
where $x_{j}$ are states of contributing neurons, $W_{ij}$ synaptic weights, $b_{i}$ is a bias term and $\sigma$ is nonlinear activation function that models synaptic transmission of signal \cite[Section 1.4.3]{hoppensteadt2012weakly}. Equation~(\ref{eq:threshold_unit_dynamics}) represents a specific instance of an additive model, which has a direct connection to the Hopfield associative memory framework \cite{grossberg1988nonlinear}, \cite{hopfield1984neurons}. The rationale behind the threshold unit is schematically captured in the left panel of Figure~\ref{fig:three_types_of_neurons}. In the context of deep learning, the continuous dynamics of the threshold unit are often approximated by a discrete update rule $\sum_{j}W_{ij} \sigma\left(x_{j}\right)$. While the temporal evolution differs, the fundamental interpretation in terms of weighted input and nonlinear activation remains largely consistent.

However, threshold units inherently fail to capture crucial aspects of neuronal interaction. As illustrated in Figure~\ref{fig:three_types_of_neurons}, they are insensitive to the relative timing, or phase difference, of neuronal spikes. This phase information is critical because even if neurons exhibit high average activity (large $\sigma(x)$,  effective communication can be hindered if presynaptic spikes arrive during the postsynaptic neuron's refractory period. Such phase-dependent interactions, along with other effects related to frequency modulation, compel alternative simplified models. For instance, Izhikevich demonstrated that under certain conditions, the dynamics of periodically firing pulse-coupled neurons can be effectively described by the Kuramoto equation \cite[Theorem 1, Section V.D, equation (17)]{izhikevich1999weakly}. The Kuramoto neuron's state is represented by a phase $\phi_i\in[0, 2\pi]$. In the absence of coupling, each neuron oscillates with its intrinsic frequency $\omega_i$nd their interaction is governed by the Kuramoto equation:
\begin{equation}
	\label{eq:Kuramoto_unit_dynamics}
	\dot{\phi}_{i} = \omega_{i} + \sum_{j}w_{ij}\sin(\phi_{j} - \phi_{i})
\end{equation}
Kuramoto neuron, an example of a more general oscillatory neuron, is graphically described in the right panel of Figure~\ref{fig:three_types_of_neurons}.

The artificial Kuramoto oscillatory neuron, recently introduced in \cite{miyato2024artificial}, represents a natural extension of the standard Kuramoto model (Equation~(\ref{eq:Kuramoto_unit_dynamics})). Its dynamics is based on the generalized Kuramoto equation proposed by Lohe \cite{lohe2009non}:
\begin{equation}
	\label{eq:Kuramoto_oscillatory_neuron dynamics}
	\dot{\boldsymbol{\mu}}_{i} = \boldsymbol{\Omega}_{i}\boldsymbol{\mu}_i + \left(\boldsymbol{I} - \boldsymbol{\mu}_i\boldsymbol{\mu}_i^{\top}\right)\sum_{j}w_{ij}\boldsymbol{\mu}_j,
\end{equation}
where $\boldsymbol{\mu}_i \in \mathbb{R}^{D+1}$, $\boldsymbol{\Omega}_{i}$ is a skew-symmetric matrix. It can be readily observed that $\boldsymbol{\mu}_i^{\top} \boldsymbol{\mu}_i$ remains constant, and given that $\boldsymbol{\Omega}_{i}$ is skew symmetric (describe rotation), $\boldsymbol{\mu}_i$ evolves on the surface of the sphere. Furthermore, with initial conditions satisfying $\boldsymbol{\mu}_i^{\top} \boldsymbol{\mu}_i = 1$ the generalized Kuramoto equation (Equation~(\ref{eq:Kuramoto_oscillatory_neuron dynamics})) simplifies to the standard Kuramoto model (Equation~(\ref{eq:Kuramoto_unit_dynamics})) when $D=1$.

Thus, the generalized Kuramoto model (Equation~(\ref{eq:Kuramoto_oscillatory_neuron dynamics})) encompasses the standard Kuramoto neuron (Equation~(\ref{eq:Kuramoto_unit_dynamics})) as a specific case when $D=1$. Even when $D > 1$ the generalized Kuramoto model retains biological plausibility. Artificial neurons are often interpreted as representing local populations of biological neurons rather than individual cells \cite[Section 1.4.2]{hoppensteadt2012weakly}. While this population averaging is inherent in additive models without altering their form, for oscillatory models, it becomes plausible to represent a neuronal population with an oscillator possessing multiple intrinsic frequencies. In the generalized Kuramoto equation, these frequencies are linked to the eigenvalues of the skew-symmetric matrix $\boldsymbol{\Omega}_{i}$.\footnote{Recall, that skew-symmetric matrix is unitary equivalent to block diagonal matrix with either $1\times 1$ zero block or $2\times2$ block $\begin{pmatrix} 0 & -\lambda_i \\ \lambda_i & 0\end{pmatrix}$. The later one describes rotation with frequency $\lambda$.}

As we have seen, both threshold units and artificial oscillatory neurons find justification from biological principles. Our analysis reveals that they correspond to distinct aspects of neural coding: threshold units primarily model averaged neuronal activity, disregarding temporal phase, while oscillatory neurons specifically capture interactions arising from the phase difference between neural oscillations. The co-existence of both threshold units and oscillatory neurons within a single network holds the promise of capturing a richer repertoire of neural dynamics than either model can achieve in isolation. In the subsequent two sections, we will leverage the general formalism of associative memory to introduce such a coupled model, bridging these distinct modes of neural interaction.

\section{Two memory models}
\label{section:Two memory models}
A canonical way to model associative memory was formulated in \cite{hopfield1984neurons}. One constructs autonomous dynamical system
\begin{equation}
	\label{eq:abstract_associative_memory}
	\dot{\boldsymbol{x}} = \boldsymbol{F}\left(\boldsymbol{x}\right)
\end{equation}
with Lyapunov function $E$ non-increasing on trajectories of dynamical system
\begin{equation}
	\label{eq:energy_property}
	\dot{E}\left(\boldsymbol{x}\right) = \left(\frac{\partial E\left(\boldsymbol{x}\right)}{\partial \boldsymbol{x}}\right)^{\top}\dot{\boldsymbol{x}} =  \left(\frac{\partial E\left(\boldsymbol{x}\right)}{\partial \boldsymbol{x}}\right)^{\top}\boldsymbol{F}\left(\boldsymbol{x}\right)\leq 0.
\end{equation}
For sufficiently good energy functions, condition~(\ref{eq:energy_property}) ensures that trajectories end up at particular steady state $\boldsymbol{x}^{\star}:\boldsymbol{F}\left(\boldsymbol{x}^{\star}\right) = 0$. Each steady state has basin of attraction, the set of initial conditions that converge to $\boldsymbol{x}^{\star}$, $\mathcal{D}_{\boldsymbol{x}^{\star}}=\left\{\boldsymbol{y}:\boldsymbol{x}(0)=\boldsymbol{y},\,\boldsymbol{x}(\infty)=\boldsymbol{x}^{\star},\dot{\boldsymbol{x}} = \boldsymbol{F}\left(\boldsymbol{x}\right)\right\}$.

Dynamical system~(\ref{eq:abstract_associative_memory}) with these properties has a natural interpretation: steady states $\boldsymbol{x}^{\star}$ are stored memories and their number is the memory capacity; attractor $\mathcal{D}_{\boldsymbol{x}^{\star}}$ of steady state $\boldsymbol{x}^{\star}$ is a set of all stimuli that lead to the recall of $\boldsymbol{x}^{\star}$.

Associative memory can also be understood as a productive computation platform. It was applied to many practical machine learning tasks including denoising \cite{salvatori2021associative}, \cite{hoover2024energy}, clustering \cite{maetschke2014characterizing}, \cite{kandali2021new}, \cite{saha2023end}, classification \cite{widrich2020modern}, \cite{ramsauer2020hopfield},  representation learning \cite{hoover2022universal}, \cite{hoover2024energy}.

\subsection{Hopfield associative memory}
\label{subsection:Hopfield}
We consider the continuous Hopfield memory model introduced in \cite{krotov2020large}. More precisely, we use the dense Hopfield layer from \cite{krotov2021hierarchical}. Both constructions are based on the Lagrange function $L\left(\boldsymbol{x}\right)$ that is a scalar function of a state variable whose gradient equals activation function. The dynamical system of the model reads
\begin{equation}
	\label{eq:Hopfield_memory}
	\dot{\boldsymbol{x}} = \boldsymbol{W} \frac{\partial L\left(\boldsymbol{x}\right)}{\partial \boldsymbol{x}} - \boldsymbol{x} + \boldsymbol{b}.
\end{equation}
In \cite{krotov2020large}, \cite{krotov2021hierarchical} authors provide the following formal characteristic of this system.
\begin{theorem}[Lyapunov function for Hopfield-Krotov associative memory]
	\label{th:Lyapunov_Hopfield}
	Let for dynamical system~(\ref{eq:Hopfield_memory}) weight matrix be symmetric $\boldsymbol{W}^{\top} = \boldsymbol{W}$, and Hessian of Lagrange function be positive-semidefinite matrix $\frac{\partial^2 L}{\partial \boldsymbol{x}^2} \geq 0$, $\boldsymbol{g}(\boldsymbol{x}) \equiv \frac{\partial L\left(\boldsymbol{x}\right)}{\partial \boldsymbol{x}}$. Energy function
	\begin{equation}
		\label{eq:Hopfield_energy}
		E_{H}(\boldsymbol{x}) = \left(\boldsymbol{x} - \boldsymbol{b}\right)^{\top}\boldsymbol{g}(\boldsymbol{x})  - L(\boldsymbol{x}) - \frac{1}{2}\boldsymbol{g}(\boldsymbol{x})^\top \boldsymbol{W}\boldsymbol{g}(\boldsymbol{x})
	\end{equation}
	is a Lyapunov function of~(\ref{eq:Hopfield_memory}) since
	\begin{equation}
		\dot{E}_{H}(\boldsymbol{x}) = -\dot{\boldsymbol{x}}^{\top} \frac{\partial^2 L(\boldsymbol{x})}{\partial \boldsymbol{x}^2} \dot{\boldsymbol{x}} \leq 0.
	\end{equation}
\end{theorem}
\begin{proof}
Proof is available in \cite{krotov2020large}, \cite{krotov2021hierarchical} and reproduced in Appendix~\ref{appendix:Lyapunov_Hopfield} for convenience.
\end{proof}
So we see that under specified conditions on $L$ and $\boldsymbol{W}$, dynamical system~(\ref{eq:Hopfield_memory}) is suitable for modelling of associative memory in a sense explained in the introduction to this section\footnote{The symmetry of weights is undesirable from biological perspective, but not from the approximation theory standpoint: universal approximation property for networks with symmetric weights is demonstrated in \cite{hu2019exploring}.}.

The construction may look rather abstract, so we provide a concrete example of Lagrange function and weight matrix $\boldsymbol{W}$. The end of this example is twofold: (i) to illustrate the approach used to build deep neural network from shallow implicit model \cite{bai2019deep}, \cite{krotov2021hierarchical}; (ii) to demonstrate that proper choice of Lagrange function may lead to neural networks with familiar activation functions.
\begin{example}[deep ReLU network]
	\label{ex:Hopfield_ReLU}
	For simplicity we will consider a network with three layers. Generalisation on deeper networks is straightforward. We consider $\boldsymbol{W}$, $\boldsymbol{b}$, $\boldsymbol{x}$ with block structure
	\begin{equation*}
		\boldsymbol{W} = 
		\begin{pmatrix}
			0 & \boldsymbol{W}_{12} & 0 \\
			\boldsymbol{W}_{12}^{\top} & 0 & \boldsymbol{W}_{23} \\
			0 & \boldsymbol{W}_{23}^{\top} & 0
		\end{pmatrix},\,
		\boldsymbol{b} = 
		\begin{pmatrix}
			\boldsymbol{b}_1 \\
			\boldsymbol{b}_2 \\
			\boldsymbol{b}_3
		\end{pmatrix},\,
		\boldsymbol{x} = 
		\begin{pmatrix}
			\boldsymbol{x}_1 \\
			\boldsymbol{x}_2 \\
			\boldsymbol{x}_3
		\end{pmatrix}.
	\end{equation*} 
	The Lagrange function that we use has additive structure $L(\boldsymbol{x}) = \frac{1}{2}\sum_{i=1}^{3}\boldsymbol{1}^{\top}\left({\sf ReLU}(\boldsymbol{x}_i)\right)^2$, where $\boldsymbol{1}$ is a vector with all components equal one. For this choice of $\boldsymbol{W}$ and Lagrange function, steady-state reads
	\begin{equation*}
		\begin{split}
			&\boldsymbol{x}_1 = \boldsymbol{W}_{12}{\sf ReLU}(\boldsymbol{x}_2) + \boldsymbol{b}_1,\,\boldsymbol{x}_3 = \boldsymbol{W}_{32}{\sf ReLU}(\boldsymbol{x}_2) + \boldsymbol{b}_3;\\
			&\boldsymbol{x}_2 = \boldsymbol{W}_{21}{\sf ReLU}(\boldsymbol{x}_1) + \boldsymbol{W}_{23}{\sf ReLU}(\boldsymbol{x}_3)  + \boldsymbol{b}_2.
		\end{split}
	\end{equation*}
	So we see that steady-state has a structure of neural network with three layers and symmetric feedback connections.
\end{example}

Techniques demonstrated in this section were used to derive energy function for many interesting architectures including dense associative memory \cite{krotov2020large} modern Hopfield network \cite{ramsauer2020hopfield}, convolutional neural networks \cite{krotov2021hierarchical}, \cite{hoover2022universal}, energy transformer \cite{hoover2024energy}, MLP-mixer \cite{tang2021remark}, diffusion models \cite{ambrogioni2023search}, \cite{pham2024memorization}.

\subsection{Kuramoto associative memory}
\label{subsection:Kuramoto}
There exist many models of associative memory based on synchronisation of oscillators, e.g., \cite{izhikevich1999weakly}, \cite{hoppensteadt1999oscillatory}, \cite{maffezzoni2015oscillator}. A novel model of associative memory presented here is based on a version of generalised Kuramoto model \cite{lohe2009non}, artificial oscillatory neurons \cite{miyato2024artificial} and ideas from \cite{krotov2020large}. Our main motivation is to accommodate oscillatory neurons into the general framework of associative memory to later infer possible couplings between oscillatory neurons and threshold units from the general form of global energy function.

Equations of motion for $N$ coupled oscillator reads
\begin{equation}
	\label{eq:Kuramoto_memory}
	\dot{\boldsymbol{\mu}}_{i} = \boldsymbol{\Omega}_{i}\boldsymbol{\mu}_i - \left(\boldsymbol{I} - \boldsymbol{\mu}_i\boldsymbol{\mu}_i^{\top}\right)\frac{\partial E_{K}}{\partial \boldsymbol{\mu}_i},
\end{equation}
where $E_{K}$ is a potential function analogous to Lagrange function in Hopfield memory~(\ref{eq:Hopfield_memory}). 

Kuramoto model~(\ref{eq:Kuramoto_memory}) also has Lyapunov function under certain technical conditions on $E_{K}$ and $\boldsymbol{\Omega}_{i}$ as explained below.
\begin{theorem}[Lyapunov function for Kuramoto associative memory]
	\label{th:Lyapunov_Kuramoto}
	Let for dynamical system~(\ref{eq:Kuramoto_memory}) potential $E_{K}$ depends on $\boldsymbol{\mu}_i$ only through scalar products $\boldsymbol{\mu}_i^{\top}\boldsymbol{\mu}_j$ and natural frequencies are identical $\boldsymbol{\Omega}_{i} = \boldsymbol{\Omega}$ for all oscillators. Energy function $E_{K}$ is a Lyapunov function since
	\begin{equation}
		\dot{E_{K}} = - \sum_{i}\frac{\partial E_{K}}{\partial \boldsymbol{\mu}_i} \left(\boldsymbol{I} - \boldsymbol{\mu}_i\boldsymbol{\mu}_i^{\top}\right)\frac{\partial E_{K}}{\partial \boldsymbol{\mu}_i} \leq 0.
	\end{equation}
\end{theorem}
\begin{proof}
Appendix~\ref{appendix:Lyapunov_Kuramoto}.
\end{proof}
As we show in Appendix~\ref{appendix:Lyapunov_Kuramoto},  changing to rotating frame one can reduce Kuramoto model~(\ref{eq:Kuramoto_memory}) to constrained gradient flow, under conditions of Theorem~\ref{th:Lyapunov_Kuramoto}. Given that Kuramoto associative memory~\ref{eq:Kuramoto_memory} can be used the same way as Hopfield model but with a different information encoding scheme.

The starting points on the energy landscape as well as its local minima depend only on scalar products $\boldsymbol{\mu}_i^{\top}\boldsymbol{\mu}_j$.
Given that, only scalar products $\boldsymbol{\mu}_i^{\top}\boldsymbol{\mu}_j$ should be used to recover memory vectors and also to supply input stimulus. A natural question then is: if one wants to store components of vector $v\in[-1, 1]^{K}$ in scalar products, what is the minimal number of oscillators needed? For $N$ constrained oscillators one has $N(N-1)/2$ scalar products. Unfortunately, these scalar products are not independent. It is easy to see that the worst case is realised when one needs to encode a vector $v$ with all components equal $1$, in which case all oscillators point to the same direction. So, in total, with $N$ oscillators on a sphere we can reliably encode $N-1$ scalars. This is explained in detail in Appendix~\ref{appendix:scalar_products_encoding}.

A special structure of energy $E_{K}$ and Kuramoto model~(\ref{eq:Kuramoto_memory}) makes it unnatural to consider hierarchical models, but Lagrange formalism from \cite{krotov2020large} allows one to construct various nonlinear coupling in dense networks with oscillatory units. Below we provide two examples with ReLU and attention activation functions.

\begin{example}[dense ReLU network]
	\label{ex:Kuramoto_ReLU}
	To define dense ReLU network we consider the following energy and coupling term
	\begin{equation*}
		\begin{split}
			&E_{K} = \sum_{i=1}^{N-1}\sum_{j>i}\left(\frac{1}{2}W_{ij}\left({\sf ReLU}\left(\boldsymbol{\mu}_{i}^{\top}\boldsymbol{\mu}_{j}\right)\right)^2 + b_{ij}\boldsymbol{\mu}_{i}^{\top}\boldsymbol{\mu}_{j}\right),\\
			&\frac{\partial E_{K}}{\partial \boldsymbol{\mu}_{i}} = \sum_{j\neq i}\left(W_{ij}{\sf ReLU}\left(\boldsymbol{\mu}_{i}^{\top}\boldsymbol{\mu}_{j}\right) + b_{ij}\right)\boldsymbol{\mu}_{j}.
		\end{split}
	\end{equation*}
	The form of a coupling is clearly related to the Hopfield network with ReLU activation functions. The main difference is that for Hopfield network degrees of freedom $\boldsymbol{x}$ can be processed by nonlinear function directly, and in the Kuramoto memory model this can be done only with scalar products.
\end{example}

\begin{example}[dense attention network]
	\label{ex:Kuramoto_attention}
	Similarly one can construct attentive interaction
	\begin{equation*}
		\begin{split}
			&E_{K} = \log\left( \sum_{i=1}^{N-1}\sum_{j>i}\exp\left(W_{ij} \boldsymbol{\mu}_{i}^{\top}\boldsymbol{\mu}_{j}\right)\right) + \sum_{i=1}^{N-1}\sum_{j>i} b_{ij} \boldsymbol{\mu}_{i}^{\top}\boldsymbol{\mu}_{j}\\
			&\frac{\partial E_{K}}{\partial \boldsymbol{\mu}_{i}} = \sum_{j\neq i}\left(W_{ij}\frac{\exp\left(W_{ij} \boldsymbol{\mu}_{i}^{\top}\boldsymbol{\mu}_{j}\right)}{\sum_{i=1}^{N-1}\sum_{j>i}\exp\left(W_{ij} \boldsymbol{\mu}_{i}^{\top}\boldsymbol{\mu}_{j}\right)} + b_{ij}\right)\boldsymbol{\mu}_{j}.
		\end{split}
	\end{equation*}
	This model is analogous to Modern Hopfield Network defined in \cite{ramsauer2020hopfield}.
\end{example}

The result we used here relies heavily on the assumption $\boldsymbol{\Omega}_{i} = \boldsymbol{\Omega}$. This means all natural frequencies are the same so frequency synchronisation is trivial. When frequencies are different but do not deviate too much from the mean frequency, the generalised Kuramoto model still converges to synchronisation state \cite{markdahl2021almost}. When frequency terms differ significantly, synchronisation is not possible. This provides an interesting opportunity to drive the system away from local energy minima without remodelling of potential, countering the argument against attractor-based formalism presented in \cite{koch2024biological}.

The assumption $\boldsymbol{\Omega}_{i} = \boldsymbol{\Omega}$ for $D\gg1$ is far less restrictive than $\omega_i = \omega$ in classical Kuramoto model~(\ref{eq:Kuramoto_unit_dynamics}), especially when learning dynamics is considered. When $D\gg1$ we effectively allow each neuron to choose from a large number of potentially different frequencies. Given that, we expect the Kuramoto associative memory model to be more expressive when $D$ increases. This is partially confirmed in Section~\ref{section:Numerics}.

\section{Hopfield-Kuramoto associative memory}
\label{section:Hopfield-Kuramoto}
To define a joint Hopfield-Kuramoto memory model we consider $N$ neurons that have $D$ oscillatory and $1$ scalar degrees of freedom. The equations of motion are those of Hopfield~(\ref{eq:Hopfield_memory}) and Kuramoto~(\ref{eq:Kuramoto_memory}) models but with additional interaction terms
\begin{equation}
	\label{eq:Hopfield_Kuramoto_memory}
	\begin{split}
		&\dot{x_{i}} = \sum_{j}W_{ij} \left(\frac{\partial L\left(\boldsymbol{x}\right)}{\partial \boldsymbol{x}}\right)_{j} - x_i + b_i + \frac{1}{\kappa_{H}}\Delta x_{i}(\boldsymbol{x}, \boldsymbol{\mu}),\\
		&\dot{\boldsymbol{\mu}}_{i} = \boldsymbol{\Omega}_{i}\boldsymbol{\mu}_i + \left(\boldsymbol{I} - \boldsymbol{\mu}_i\boldsymbol{\mu}_i^{\top}\right)\left(-\frac{\partial E_{K}}{\partial \boldsymbol{\mu}_i} + \frac{1}{\kappa_{K}}\delta \boldsymbol{\mu}_{i}(\boldsymbol{x}, \boldsymbol{\mu})\right),\\
	\end{split}
\end{equation}
where we rewrite Hopfield model~(\ref{eq:Hopfield_memory}) in its scalar form to align it better with Kuramoto model, and introduce coupling constants $\kappa_{H},\kappa_{K}>0$ that control interaction strength.

The main difficulty is to introduce coupling terms that still allow us to define Lyapunov function. An appropriate strategy is to modify existing Hopfield and Kuramoto energy functions to introduce interaction. Kuramoto energy function should depend only on scalar products $\boldsymbol{\mu}_{i}^\top\boldsymbol{\mu}_{j}$ and this is the only constraint. On the other hand, the Hopfield energy function~(\ref{eq:Hopfield_energy}) has a fairly restricted form. Part where scalar products $\boldsymbol{\mu}_{i}^\top\boldsymbol{\mu}_{j}$ fit naturally is the term with synaptic weights $W_{ij}$. This suggest us to replace the last term in Hopfield energy~(\ref{eq:Hopfield_energy}) with
\begin{equation*}
	 -  \sum_{ij}W_{ij}\left(\boldsymbol{\mu}_i^{\top}\boldsymbol{\mu}_j\right)\left(\frac{\partial L\left(\boldsymbol{x}\right)}{\partial \boldsymbol{x}}\right)_i \left(\frac{\partial L\left(\boldsymbol{x}\right)}{\partial \boldsymbol{x}}\right)_{j},
\end{equation*}
which corresponds to the multiplication of synaptic weights on scalar products $\left(\boldsymbol{\mu}_i^{\top}\boldsymbol{\mu}_j\right)$. Clearly, these manipulations satisfy constraints of Kuramoto energy since only scalar products are involved. Besides that they also lead to valid Hopfield energy, since weights remain symmetric and other parts of energy function~(\ref{eq:Hopfield_energy}) are not modified.

Following similar lines we derived more general interaction terms with the Lyapunov function presented below.
\begin{theorem}[Lyapunov function for Hopfield-Kuramoto associative memory]
\label{th:Lyapunov_Hopfield_Kuramoto}
Let $\boldsymbol{g}(\boldsymbol{x}) \equiv\frac{\partial L\left(\boldsymbol{x}\right)}{\partial \boldsymbol{x}}$, $G_{ij} = G_{ji}$ and $\chi_{ij} = \chi_{ji}$ be smooth scalar function and Hessian of Lagrange function be positive-semidefinite matrix $\frac{\partial^2 L}{\partial \boldsymbol{x}^2} \geq 0$ and $\boldsymbol{\Omega}_i=\boldsymbol{\Omega}$, $\kappa_{H}$ and $\kappa_{K}$ are positive coupling constants. Energy function 
\begin{equation}
	\label{eq:Hopfield_Kuramoto_energy}
	E_{HK}\left(\boldsymbol{x}, \boldsymbol{\mu}\right) = \kappa_{H}E_{H}\left(\boldsymbol{x}\right) + \kappa_{K} E_{K}\left(\boldsymbol{\mu}\right) \\ - \frac{1}{2}\sum_{ij}G_{ij}\left(g(\boldsymbol{x})_{i} g(\boldsymbol{x})_{j}\right)\chi_{ij}(\boldsymbol{\mu}_{i}^{\top}\boldsymbol{\mu}_{j})
\end{equation}
is Lyapunov function for joint Kuramoto-Hopfield model~(\ref{eq:Hopfield_Kuramoto_memory}) with coupling terms
\begin{equation}
	\label{eq:Hopfield_Kuramoto_coupling}
	\begin{split}
		&\Delta x_{i}(\boldsymbol{x}, \boldsymbol{\mu}) = \sum_{j}G^{'}_{ij}\left(g(\boldsymbol{x})_{i}g(\boldsymbol{x})_{j}\right)\chi_{ij}(\boldsymbol{\mu}_{i}^{\top}\boldsymbol{\mu}_{j})g(\boldsymbol{x})_{j},\\
		&\delta\boldsymbol{\mu}_{i}(\boldsymbol{x}, \boldsymbol{\mu}) =  \sum_{j\neq i}G_{ij}\left(g(\boldsymbol{x})_{i} g(\boldsymbol{x})_{j}\right)\chi^{'}_{ij}(\boldsymbol{\mu}_{i}^{\top}\boldsymbol{\mu}_{j}) \boldsymbol{\mu}_{j},
	\end{split}
\end{equation}
where $G^{'}_{ij}$ and $\chi^{'}_{ij}$ are derivatives of $G_{ij}$ and $\chi_{ij}$. More precisely, under conditions specified
\begin{equation}
	\label{eq:Hopfield_Kuramoto_energy_derivative}
	\dot{E}_{HK}\left(\boldsymbol{x}, \boldsymbol{\mu}\right) = -\kappa_{H}\dot{\boldsymbol{x}}^{\top} \frac{\partial^2 L(\boldsymbol{x})}{\partial \boldsymbol{x}^2} \dot{\boldsymbol{x}} \\- \frac{1}{\kappa_{K}}\sum_{i}\frac{\partial E_{HK}}{\partial \boldsymbol{\mu}_i} \left(\boldsymbol{I} - \boldsymbol{\mu}_i\boldsymbol{\mu}_i^{\top}\right)\frac{\partial E_{HK}}{\partial \boldsymbol{\mu}_i} \leq 0.
\end{equation}
\end{theorem}
\begin{proof}
Appendix~\ref{appendix:Hopfield_Kuramoto}.
\end{proof}
Coupling~(\ref{eq:Hopfield_Kuramoto_coupling}) allows for a rich set of models with various control of threshold units by oscillatory units and vice versa. Using formalism of Lagrange functions from \cite{krotov2020large} it is possible to introduce Hopfield-Kuramoto architectures with convolution layers, attention mechanism, normalisation, pooling, etc. Several concrete architectures are available in Appendix~\ref{appendix:Models}. Below we build a quantitative example that shows how oscillatory degrees of freedom can realise multiplexing.
\begin{example}[multiplexing with oscillatory units]
	\label{ex:Multiplexing}
	We will demonstrate that in the joint Hopfield-Kuramoto memory model the same set of threshold units can be reused in two different computations. Which computation is invoked is determined by oscillatory units. Let $G_1$ and $G_2$ be two sets of indices $G_1\cap G_2 = \emptyset$. Neurons from $G_2$ interact with neurons from $G_1$ with synaptic weights $W_{ij}, i\in G_1, j\in G_2$. We select coupling constants $\kappa_{K}\gg 1$ and $\kappa_{H} = 1$. With this choice dynamics for oscillatory units reduced to uncoupled Kuramoto memory~(\ref{eq:Kuramoto_memory}), but threshold units remain coupled to oscillatory units with $\Delta x_i$. 
	
	Suppose we select energy $E_{K}$ in such a way to for input $\boldsymbol{\mu}_{i}(0) = \boldsymbol{\mu}_{i}^{(A)}$ oscillatory neurons from $G_1$ and $G_2$ converges to aligned state $\boldsymbol{\mu}_{i}^{\top}(\infty)\boldsymbol{\mu}_{j}(\infty) = 1$ and for input $\boldsymbol{\mu}_{i}(0) = \boldsymbol{\mu}_{i}^{(B)}$ they converge to the orthogonal state $\boldsymbol{\mu}_{i}^{\top}(\infty)\boldsymbol{\mu}_{j}(\infty) = 0$. Recall that $E_{K}$ is not constrained, so as long as inputs $A$ and $B$ are sufficiently distinct it is possible to construct such energy function.
	
	Choosing $\Delta x_i = \sum_{j}\left(\widetilde{W}_{ij} - W_{ij}\right)\boldsymbol{\mu}_{i}^{\top}\boldsymbol{\mu}_{j}g(\boldsymbol{x})_i$ as correction term we observe that for input $A$ the effect of correction term is the change of weights from $W_{ij}$ to $\widetilde{W}_{ij}$ for interaction of neuron in $G_1$ and $G_2$. This means selected threshold units perform two different functions based on the state of oscillatory neurons. In a similar way multiplexing with more than two functions can be achieved with additional interaction terms.
\end{example}

Multiplexing example suggests that oscillatory neurons can remodel energy landscape of Hopfield associative memory during recall. This means memories stored in the Hopfield-Kuramoto network explicitly depend on the input stimulus. Formally, Hopfield model implement mapping $\boldsymbol{x}(0) \rightarrow \left\{\boldsymbol{x}_1,\dots,\boldsymbol{x}_M\right\}$, where $M$ is memory capacity, i.e., it selects among $M$ stored memories. If we suppose for simplicity that initial state of oscillatory units is a function of $\boldsymbol{x}(0)$, Hopfield-Kuramoto model realises mapping $\boldsymbol{x}(0) \rightarrow \left\{\boldsymbol{x}_1(\boldsymbol{x}(0)),\dots,\boldsymbol{x}_{M(\boldsymbol{x}(0) )}(\boldsymbol{x}(0) )\right\}$ where both capacity and stored memories are functions of $\boldsymbol{x}(0)$. We will see numerical example of memory edit in Section~\ref{section:Numerics}.

\section{Interpretations of binding mechanism}
\label{section:Binding}
While the general form of the coupling mechanism (Equation~(\ref{eq:Hopfield_Kuramoto_coupling})) offers flexibility, its interpretation can be nuanced. To gain clearer insights, and building upon Example~\ref{ex:Multiplexing}, we now analyze specific choices for the interaction terms and explore their connections to established mechanisms in standard deep learning architectures.

Let us first consider a specific choice of interaction functions: $G_{ij}(a) = \widetilde{W}_{ij} a$, $\chi_{ij}(a) = a$, $W_{ij} = 0$, $E_{K} = 0$, $\kappa_{H} = \kappa_{K} = 1$. This particular selection leads to the following Hopfield-Kuramoto model:
\begin{equation}
	\label{eq:low_rank_coupling}
	\begin{split}
		&\dot{x}_{i} = \sum_{j} \left(\boldsymbol{\mu}_{i}^{\top} \boldsymbol{\mu}_{j} \widetilde{W}_{ij}\right) \boldsymbol{g}(\boldsymbol{x})_{j} - x_{i} + b_{i},\\
		&\dot{\boldsymbol{\mu}}_{i} = \boldsymbol{\Omega}_{i}\boldsymbol{\mu}_i + \left(\boldsymbol{I} - \boldsymbol{\mu}_i\boldsymbol{\mu}_i^{\top}\right) \sum_{j}\left(\boldsymbol{g}(\boldsymbol{x})_{i} \boldsymbol{g}(\boldsymbol{x})_{j} \widetilde{W}_{ij} \right) \boldsymbol{\mu}_{j}.
	\end{split}
\end{equation}
Observing the equation for the threshold units, we find that the term $\boldsymbol{\mu}_{i}^{\top} \boldsymbol{\mu}_{j}$ acs as a neuron-specific gating mechanism on the influence of the threshold units $i$ and $j$. This bears resemblance to the gating mechanisms employed in LSTM and GRU networks \cite{hochreiter1997long}, \cite{cho2014properties}. Specifically, when the inner product $\boldsymbol{\mu}_{i}^{\top} \boldsymbol{\mu}_{j}$ is zero, the interaction between threshold units $i$ and $j$ is completely suppressed, even if their individual activities $\boldsymbol{g}(\boldsymbol{x})_{i}$, $\boldsymbol{g}(\boldsymbol{x})_{j}$ are not zero. Biologically, this suppression could be analogous to a scenario where an action potential fails to evoke a response due to the postsynaptic neuron being in its refractory phase. Even if the average activity of connected neurons is high ($g(x) \neq 0$), effective interaction might be blocked if a presynaptic action potential arrives during the postsynaptic neuron's refractory period, a state where the neuron is temporarily less responsive. Conversely, the synchronization of oscillations between two neurons is contingent on their corresponding threshold unit activities being non-zero ($\boldsymbol{g}(\boldsymbol{x})_{i}\neq 0, \boldsymbol{g}(\boldsymbol{x})_{j}\neq 0$), suggesting that the threshold units modulate the oscillatory interactions.

Now, let us consider another choice for the interaction terms: $G_{ij}(a) = a$, $\chi_{ij}(a) = a$, $E_{K} = -\sum_{i > j}\boldsymbol{\mu}_{i}^{\top}  \boldsymbol{\mu}_{j}R_{ij}$, $\kappa_{H} = \kappa_{K} = 1$. This selection yields the following coupled dynamics:
\begin{equation}
	\begin{split}
	&\dot{x}_i = \sum_{j}\left(W_{ij} + \boldsymbol{\mu}_{i}^{\top} \boldsymbol{\mu}_{j}\right)\boldsymbol{g}\left(\boldsymbol{x}\right)_{j} - x_{i} + b_{i},\\
	&\dot{\boldsymbol{\mu}}_{i} = \boldsymbol{\Omega}_{i}\boldsymbol{\mu}_i + \left(\boldsymbol{I} - \boldsymbol{\mu}_i\boldsymbol{\mu}_i^{\top}\right) \sum_{j}\left(R_{ij} + \boldsymbol{g}(\boldsymbol{x})_{i} \boldsymbol{g}(\boldsymbol{x})_{j} \right)\boldsymbol{\mu}_{j}.
	\end{split}
\end{equation}
Here, we observe that the weight matrix $\boldsymbol{R}$ governing the interaction between oscillatory neurons undergoes a rank-$1$ update based on the outer product of the threshold unit activities: $\boldsymbol{R} \leftarrow \boldsymbol{R} + \boldsymbol{g}(\boldsymbol{x})\boldsymbol{g}(\boldsymbol{x})^\top$. Similarly, the weight matrix of threshold units subjects to a rank-$m$ correction, where $m \leq D+1$, determined by the dimensionality of the oscillatory neuron state $\boldsymbol{\mu}_{i} \in \mathbb{R}^{D+1}$. The rank-$m$ correction arises from the term $\boldsymbol{\mu}_{i}^\top \boldsymbol{\mu}_{j}$. Defining a matrix $\boldsymbol{G}$ with entries $G_{ij} \equiv \boldsymbol{\mu}_{i}^{\top} \boldsymbol{\mu}_{j}$ we recognise $\boldsymbol{G}$ as a Gram matrix of vectors $\boldsymbol{\mu}_{i}$. Since $\boldsymbol{G}$ is a Gram matrix formed from $N$ vectors in $\mathbb{R}^{D+1}$, it is positive semi-definite with rank at most $D+1$ (assuming $N>D$). This property implies the existence of a matrix $\boldsymbol{A}\left(\boldsymbol{\mu}\right) \in \mathbb{R}^{N\times k}$ such that $\boldsymbol{G}$ can be factorised as $\boldsymbol{G} = \boldsymbol{A}(\mu)\boldsymbol{A}(\mu)^{\top}$.\footnote{An explicit way to find $\boldsymbol{A}$ is to use Cholesky decomposition.} Thus, this specific form of coupling leads to two low-rank update rules for the weight matrices:
\begin{equation}
	\label{eq:low_rank_corrections}
	\boldsymbol{R} \leftarrow \boldsymbol{R} + \boldsymbol{g}(\boldsymbol{x})\boldsymbol{g}(\boldsymbol{x})^\top,\,\boldsymbol{W} \leftarrow \boldsymbol{W} + \boldsymbol{A}(\mu)\boldsymbol{A}(\mu)^{\top}.
\end{equation}
Interestingly, similar low-rank weight updates are employed in deep learning, notably in the Low-Rank Adaptation (LoRA) method \cite{hu2022lora}, a computationally efficient technique for fine-tuning large language models. A key distinction from standard LoRA is that the factors in our low-rank corrections (Equation~(\ref{eq:low_rank_corrections})) are functions of the network's state ($\boldsymbol{\mu}_i$, $x_i$), i.e., time-dependent. This bears similarity to the concept of ``fast weights,'' which were also proposed as low-rank, time-dependent modifications to weight matrices \cite{ba2016using}, \cite{schmidhuber1993reducing}.

From a biological perspective, the weight updates in Equation~(\ref{eq:low_rank_corrections}) can be interpreted as a form of Hebbian learning, a principle widely used for pattern storage in Hopfield networks \cite{hopfield1982neural}, \cite[Definition 2]{albanese2024hebbian}. In this context, the oscillatory degrees of freedom can potentially contribute up to $D+1$ new patterns to those encoded in the threshold unit weight matrix $\boldsymbol{W}$, while the threshold unit activities can, in turn, influence at most one pattern stored in the oscillatory unit weight matrix $\mathbb{R}$. Given the time-dependent nature of the correction terms in Equation~(\ref{eq:low_rank_corrections}), we anticipate a significant influence on the basins of attraction of the network's dynamics. However, as $\boldsymbol{\mu}_{i}^\top \boldsymbol{\mu}_{j}$ and $\boldsymbol{x}$ are expected to converge o steady-state values, the effective rank of these corrections at equilibrium is expected to be at most $D+1$ and $1$ respectively. Consistent with the multiplexing phenomenon illustrated in Example~\ref{ex:Multiplexing}, these weight corrections exhibit an explicit dependence on the initial conditions of the network.

To empirically validate the implications of this low-rank coupling mechanism (Equation~(\ref{eq:low_rank_corrections})) and the interpretations discussed, we will present results from simple learning problems in Section~\ref{section:Numerics}.

\section{Numerical illustration of coupling mechanism}
\label{section:Numerics}

To illustrate the effect of coupling numerically we select the following model
\begin{equation}
	\label{eq:exp_model}
	\begin{split}
	&\dot{x}_i = \sum_{j}\left(W_{ij} + \kappa\boldsymbol{\mu}_{i}^{\top} \boldsymbol{\mu}_{j}\right)\boldsymbol{\text{ReLU}}\left(\boldsymbol{x}\right)_{j} - x_{i} + b_{i},\\
	&\dot{\boldsymbol{\mu}}_{i} = \left(\boldsymbol{I} - \boldsymbol{\mu}_i\boldsymbol{\mu}_i^{\top}\right) \sum_{j}\left(R_{ij} + \kappa\text{ReLU}(\boldsymbol{x})_{i} \text{ReLU}(\boldsymbol{x})_{j} \right)\boldsymbol{\mu}_{j}\\
	&\quad\,+ \left(\boldsymbol{I} - \boldsymbol{\mu}_i\boldsymbol{\mu}_i^{\top}\right) \sum_{j} S_{ij}\text{ReLU}\left(\boldsymbol{\mu}_{i}^\top \boldsymbol{\mu}_{j}\right) \boldsymbol{\mu}_{j} + \boldsymbol{\Omega}\boldsymbol{\mu}_i ,
	\end{split}
\end{equation}
where both $\boldsymbol{S}$ and $\boldsymbol{R}$ are sparse matrices with sparsity pattern takes from adjacency matrix of Watts-Strogatz small world network with constant probability of edge rewriting $p=0.1$ and variable number of neighbours $k$ \cite{watts1998collective}. Sparse matrices were chosen to reduce memory consumption and a small world network was selected because of its biological plausibility. Model~(\ref{eq:exp_model}) is a particular version of low-rank coupling~(\ref{eq:low_rank_coupling}) with slightly altered $E_{K}$ intended to increase expressivity of subnetwork with oscillatory neurons.

\begin{figure}[t!]
	\centering
	\begin{tikzpicture}
		\node at (0, 0) (empty_node)  {};
		
		\draw[very thick] (0.4,0.87) rectangle (8.0,3.1);
		
		
		\begin{scope}[shift={(5.0, 1.97)}]
			\draw[thick,-] (-1, -0.83) -- (1, -0.83);
			\draw[thick,-] (-0.53, -1) -- (-0.53, 1);
			\draw[scale=0.15, domain=-6.5:6.5, smooth, variable=\x, black, very thick] plot ({\x}, {10 / (1 + exp(-\x)) - 5});
			\draw[] (-1.35, 0.85) node[right] {$\sigma(x)$};
			\filldraw[black] (0.3, 0.57) circle (2pt);
			\draw[very thick,-] (0.3, -0.83) -- (0.3, 0.56);
		\end{scope}
		
		\node at (1.5, 2.0)  (image) {\includegraphics[width=0.12\textwidth,trim={3cm 1.3cm 3cm 1.3cm},clip]{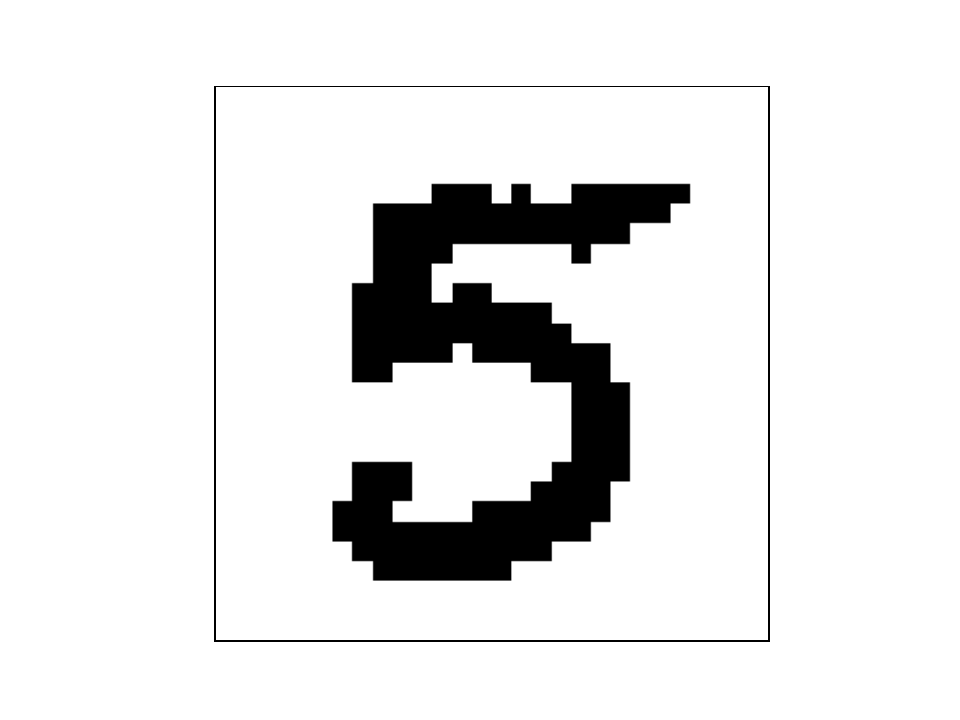}};
		
		\begin{scope}[shift={(5.0,-1.5)}]
			\draw[thick,-] (-1, 0) -- (1, 0);
			\draw[thick,-] (0, -1) -- (0, 1) ;
			\draw[thick](0,0) circle (0.7);
			\draw[very thick] (0, 0) -- (0.4, 0.57);
			\filldraw[black] (0.4, 0.57) circle (2pt);
			\draw[very thick] (0, 0) -- (-0.6, 0.36);
			\filldraw[black] (-0.6, 0.36) circle (2pt);
			\draw[] (0.69, 0.22) node[right] {$\phi$};
		\end{scope}
		
		\node at (1.5, -1.5)  (image) {\includegraphics[width=0.12\textwidth,trim={3cm 1.3cm 3cm 1.3cm},clip]{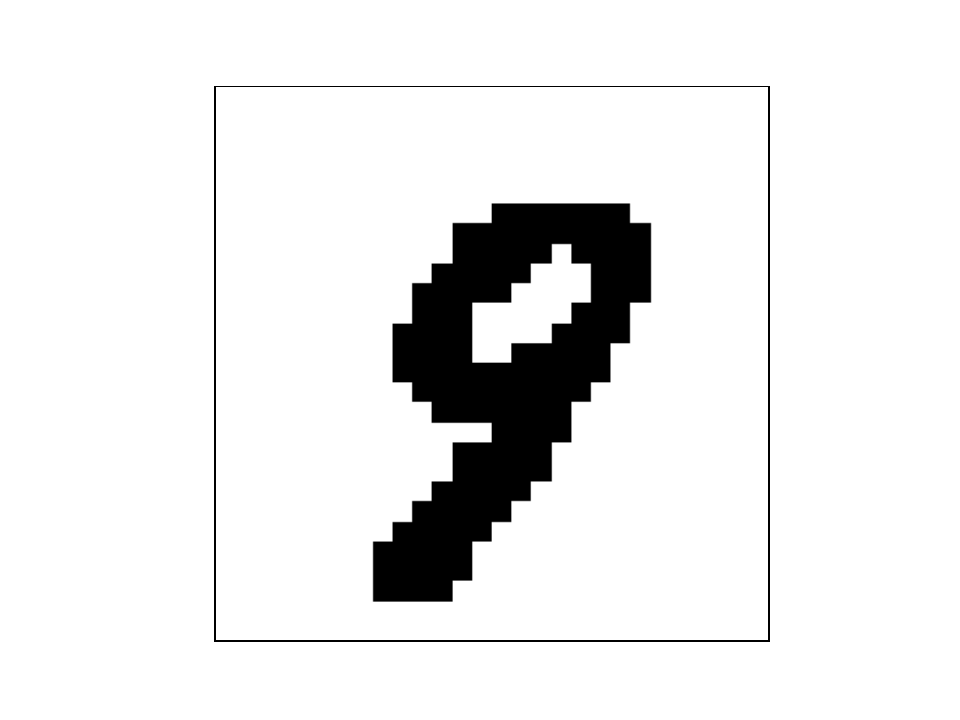}};
		
		\draw[very thick, ->] (2.8, 2.0) -- (3.7, 2.0); 
		\draw[very thick, ->] (2.8, -1.5) -- (3.7, -1.5); 
		\draw[very thick, ->] (6.2, 2.0) -- (7.1, 2.0); 
		\draw [very thick, <<.-.>>] (5.0,-0.2) to (5.0,0.6); 
		
		\node[] at (7.5, 2.0) {\resizebox{0.3cm}{!}{$5$}};
		
		\node at (7.62, 0.4)  (image) {\includegraphics[width=0.04\textwidth,trim={5cm 17cm 5cm 2cm},clip]{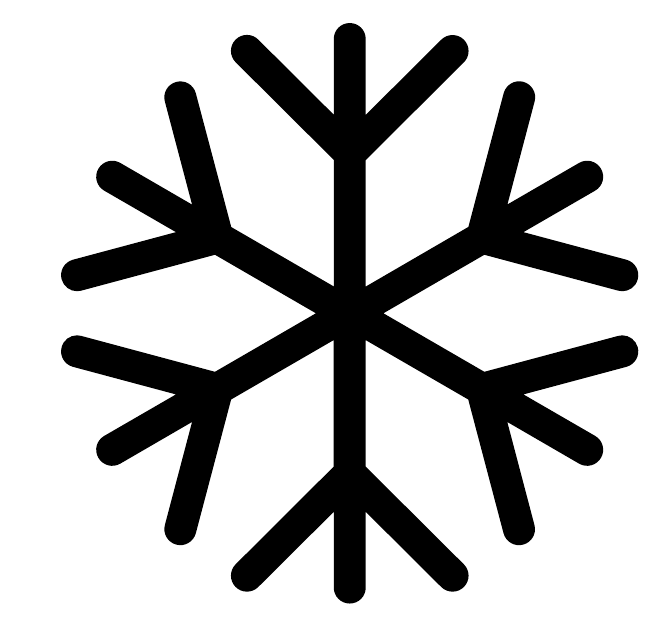}};
	\end{tikzpicture}
	\caption{Empirical validation of interaction between threshold and oscillatory units. \textbf{On the first stage} Hopfield network with threshold units is pre-trained to recognise MNIST digits and all its parameters are frozen. \textbf{On the second stage} oscillatory units are added to the network and the combined network is trained on the modified task. For example, if $1$ is present to oscillatory neurons, Hopfield network should swap labels for $0$ and $1$; if additional input is $0$, the output of Hopfield network is not modified. Learnable parameters are interaction strength (scalar) and weights of oscillatory units.}
	\label{fig:fine_tuning_experiment}
\end{figure}
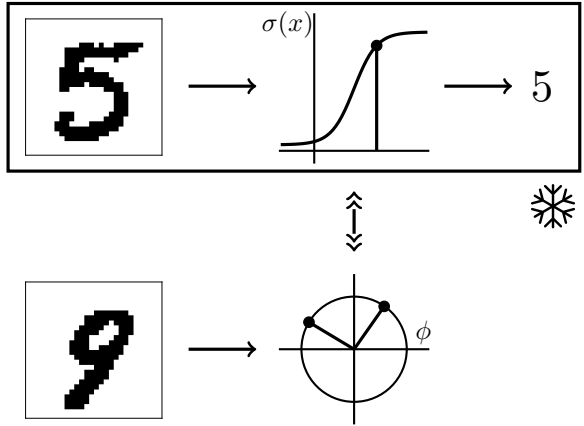

To study coupling between networks we design a special experiment protocol, illustrated in Figure~\ref{fig:fine_tuning_experiment}. First, we take $\kappa = 0$ and train Hopfield subnetwork to recognise MNIST digits \cite{lecun2010mnist}. After this stage we have trained weights $\boldsymbol{W}$ and bias term $\boldsymbol{b}$ such that starting from MNIST digit $\boldsymbol{x}(0)$, last $10$ elements of $\boldsymbol{x}(1)$ converge to class labels encoded as $+1$ for correct class and $-1$ for incorrect class. Training is done with backpropagation through time and overall the approach closely follows established techniques for learning with associative memory models \cite{krotov2016dense}, \cite{hoover2022universal}. The accuracy of Hopfield network is reported in Table~\ref{table:Hopfield_accuracy}.

\begin{table}[b!]
\caption{Performance of pre-trained Hopfield model in original MNIST dataset and its altered version: swap refers to change of label $0$ to $1$ and vice versa; conflation refers to change of label $1$ to $0$; associative version do the same but for fraction of randomly selected samples.}
\centering
\begin{tabular}{lcc}
\toprule
dataset & train accuracy & test accuracy \\
\midrule
MNIST & $98.2\%$ & $97.2\%$ \\
conflation & $87.1\%$ & $86.0\%$\\
swap & $77.4\%$ & $76.3\%$\\
associative conflation & $92.8\%$ & $91.7\%$\\
associative swap & $88.0\%$ & $86.8\%$\\
\bottomrule
\end{tabular}
\label{table:Hopfield_accuracy}
\end{table}

After the pre-training stage we modify MNIST labels and train a combined network with learnable $\kappa$ and frozen parameters of Hopfield subnetwork. In this fine-tuning stage the only trainable parameters are $\kappa$, and weights of subnetwork with oscillatory units $\boldsymbol{R}$, $\boldsymbol{S}$, $\boldsymbol{\Omega}$. The fine-tuning problem is challenging because information exchange between subnetworks comes from a very restrictive coupling term that has a single learnable parameter $\kappa$. By design of model~(\ref{eq:exp_model}) fine-tuning can be successful only when oscillatory neurons learn to converge to an appropriate low-rank correction~(\ref{eq:low_rank_corrections}) that can adapt weights $\boldsymbol{W}$ of Hopfield network to a changed task.

We consider four modification of MNIST classification problem:
\begin{itemize}
	\vspace{-0.3cm}
	\item\textbf{Swap}. Labels $0$ and $1$ are exchanged. Initial values $\boldsymbol{\mu}_i(0)$ are learnable parameters.
	\item\textbf{Conflation}. Label $1$ is changed to $0$. Initial values $\boldsymbol{\mu}_i(0)$ are learnable parameters.
	\item\textbf{Associative swap}. For each MNIST sample digits we draw random $0$ and $1$ MNIST digits with equal probability. These additional samples are used to initialise oscillatory neurons $\boldsymbol{\mu}_i(0)$. For a given sample, if a randomly selected digit (input to oscillatory units) is $0$ MNIST label is not changed, if it is $1$, labels $0$ and $1$ are exchanged.
	\item\textbf{Associative conflation}. Similar to associative swap, but label $1$ is changed to label $0$ in case additional input is digit $1$.
\end{itemize}
In other words, we exchange labels $0$ and $1$ (swap) or change label $1$ to $0$ (conflation) either for all relevant MNIST samples or only for a fraction, conditioned on the additionally selected digit $\left\{0, 1\right\}$ presented as an input (association) for oscillatory units. The accuracy of the pre-trained Hopfield network on four modified tasks is reported in Table~\ref{table:Hopfield_accuracy}. In general swap is more challenging than conflation and associative setting is harder than non-associative one. Ideally, we want fine-tuning to recover initial test accuracy of $\simeq97\%$.

\begin{table}[t!]
\caption{Results of fine-tuning for non-associative MNIST modifications.}
\centering
\begin{tabular}{llcccc}
\toprule
 & & \multicolumn{2}{c}{conflation} & \multicolumn{2}{c}{swap} \\
$k$ & $D$ & train acc. & test acc. & train acc. & test acc. \\
\midrule
$150$ & $4$ & $96.9\%$ & $96.0\%$ & $88.8\%$ & $88.2\%$\\
$250$ & $4$ & $97.9\%$ & $97.0\%$ & $96.3\%$ & $95.3\%$\\
$350$ & $4$ & $98.0\%$ & $97.1\%$ & $97.9\%$ & $96.8\%$\\
$150$ & $6$ & $97.6\%$ & $96.6\%$ & $89.5\%$ & $88.8\%$\\
$250$ & $6$ & $98.0\%$ & $97.1\%$ & $97.4\%$ & $96.5\%$\\
$350$ & $6$ & $98.1\%$ & $97.2\%$ & $98.0\%$ & $97.0\%$\\
\bottomrule
\end{tabular}
\label{table:non_associative_results}
\end{table}

We train several coupled models with frozen parameters of Hopfield network and $D=2, D=3$ and $k=150, 250, 350, 450$. Recall that $D$ is the number of oscillatory degrees of freedom for each neuron, and the larger $k$ the more trainable parameters are available in sparse matrices $\boldsymbol{R}$, $\boldsymbol{S}$. In general models with larger $D$ and $k$ should be more capable. The results of fine-tuning are available in Table~\ref{table:non_associative_results} and Table~\ref{table:associative_results}.

We can clearly see that accuracy is improved after fine-tuning. Moreover in the best cases initial $97\%$ accuracy is recovered for all tasks but associative swap, where the highest obtained accuracy is $93.6\%$. For non-associative problems we find that $D=6$ uniformly leads to better results than $D=4$. For associative problems we observe no difference between $D=4$ and $D=6$. This is likely because the bottleneck in performance comes from the expressivity of oscillatory subnetwork (controlled by $k$) because it faces more challenging tasks for this class of problems. In general we see that the results confirm that oscillatory neurons can be used to alter the content of Hopfield associative memory.

More details on numerical experiments discussed in this section are available in Appendix~\ref{appendix:Exp_description}.

\begin{table}[b!]
\caption{Results of fine-tuning for associative MNIST modifications.}
\centering
\begin{tabular}{llcccc}
\toprule
 & & \multicolumn{2}{c}{conflation} & \multicolumn{2}{c}{swap} \\
$k$ & $D$ & train acc. & test acc. & train acc. & test acc. \\
\midrule
$150$ & $4, 6$ & $92.7\%$ & $91.5\%$ & $87.8\%$ & $86.6\%$\\
$250$ & $4, 6$ & $94.0\%$ & $92.8\%$ & $88.8\%$ & $87.9\%$\\
$350$ & $4, 6$ & $97.5\%$ & $96.4\%$ & $94.7\%$ & $93.5\%$\\
$450$ & $4, 6$ & $97.9\%$ & $97.0\%$ & $94.8\%$ & $93.6\%$\\
\bottomrule
\end{tabular}
\label{table:associative_results}
\end{table}

\section{Conclusion}
In this research note, we have introduced a novel theoretically grounded coupling mechanism between artificial oscillatory neurons and traditional threshold units. This interaction term is justified based on derived joint Hopfield-Kuramoto associative memory model. Our analysis reveals that this coupling term admits several intriguing interpretations, including a form of low-rank weight correction, a multiplexing of information channels, and a dynamic gating mechanism influencing neuronal interactions. Furthermore, we have experimentally validated the capacity of oscillatory neurons to dynamically modify the stored memories within a Hopfield network through a set of carefully designed simple experiments.

These findings suggest a promising avenue for developing more sophisticated neural network architectures that leveraging the distinct computational strengths of models capturing averaged activity and oscillatory dynamics. Future work could explore the application of this coupling mechanism in deeper network architectures and investigate its potential benefits in addressing complex cognitive tasks where temporal dynamics and synchronization play a crucial role. Further theoretical analysis into the stability and capacity of the proposed Hopfield-Kuramoto model could also yield valuable insights.

\bibliography{threshold_and_oscillatory}
\bibliographystyle{plain}

\newpage
\appendix
\onecolumn
\section{Lyapunov function for Hopfield associative memory}
\label{appendix:Lyapunov_Hopfield}
As we mentioned in the main text the proof of this statement is available in \cite{krotov2020large}, \cite{krotov2021hierarchical}. Nonetheless we find it appropriate to reproduce this proof here because in Appendix~\ref{appendix:Hopfield_Kuramoto} we will need it for a more general result on the joint Hopfield-Kuramoto memory model.

We define $\boldsymbol{g}(\boldsymbol{x}) = \frac{\partial L (\boldsymbol{x})}{\partial \boldsymbol{x}}$, $\boldsymbol{\Lambda} = \frac{\partial^2 L}{\partial \boldsymbol{x}^2}$ and rewrite dynamical system~(\ref{eq:Hopfield_memory}) and energy function~(\ref{eq:Hopfield_energy}) in more explicit form
\begin{equation}
	\begin{split}
	&\dot{x}_i = \sum_{j}W_{ij}g(\boldsymbol{x})_{j} - x_{j} - b_{j},\\
	&E_{H}(\boldsymbol{x}) = \sum_{i}(x_i - b_i)g(\boldsymbol{x})_i - L(\boldsymbol{x}) - \frac{1}{2}\sum_{ij}g(\boldsymbol{x})_i W_{ij}g(\boldsymbol{x})_j.
	\end{split}
\end{equation}
Derivative of energy function reads
\begin{equation}
	\dot{E}_{H}(\boldsymbol{x}) = \sum_{i}\frac{\partial E_{H}(\boldsymbol{x})}{\partial x_{i}} \dot{x}_i,
\end{equation}
so we only need to find partial derivatives with respect to $x_i$
\begin{equation}
	\frac{\partial E_{H}(\boldsymbol{x})}{\partial x_{i}} = g(\boldsymbol{x})_i + \sum_{j}(x - b)_j \Lambda_{ij} - g(\boldsymbol{x})_i - \sum_{jk}\Lambda_{ij}W_{jk}g(\boldsymbol{x})_{k} = -\sum_{j}\Lambda_{ij} \left(\sum_{k}W_{jk}g(\boldsymbol{x})_{k} - x_{j} + b_{j}\right) = -\sum_{j}\Lambda_{ij}\dot{x}_{j}.
\end{equation}
After that we find
\begin{equation}
	\dot{E}_{H}(\boldsymbol{x}) = -\sum_{ij}\Lambda_{ij}\dot{x}_{i}\dot{x}_{j} = - \dot{\boldsymbol{x}}^{\top}\boldsymbol{\Lambda}\dot{\boldsymbol{x}} = -\dot{\boldsymbol{x}}^{\top}\frac{\partial^2 L}{\partial \boldsymbol{x}^2}\dot{\boldsymbol{x}}.
\end{equation}
By assumption, the Hessian of the Lagrange function is a positive semidefinite matrix, so $\dot{E}_{H}(\boldsymbol{x}) \leq 0$.
\section{Lyapunov function for Kuramoto associative memory}
\label{appendix:Lyapunov_Kuramoto}
We provide two distinct proofs.

The first one is based on the change of variables that remove rotation from the equation. It is somewhat longer but illustrates why it is important to have an energy function that depends on scalar products only. Besides, there we show how the Kuramoto model can be reduced to the gradient flow on the sphere.

In the second proof we proceed directly. Here one can see that the special form of energy ensure that $\boldsymbol{\Omega}_i$ vanish if $\boldsymbol{\Omega}_i = \boldsymbol{\Omega}$
\subsection{Proof by change of variables}
Matrix $\exp\left(\boldsymbol{\Omega}_i t\right)$ is orthogonal
\begin{equation}
	\left(\exp\left(\boldsymbol{\Omega}_i t\right)\right)^{\top} = \exp\left(\boldsymbol{\Omega}_i^{\top} t\right) = \exp\left(-\boldsymbol{\Omega}_i t\right) = \left(\exp\left(\boldsymbol{\Omega}_i t\right)\right)^{-1}.
\end{equation}
We use this orthogonal transformation to switch to the rotating frame
\begin{equation}
	\boldsymbol{\mu}_{i} = \exp\left(\boldsymbol{\Omega}_i t\right) \boldsymbol{\nu}_{i},\,\boldsymbol{\nu}_{i} = \exp\left(-\boldsymbol{\Omega}_i t\right)\boldsymbol{\mu}_{i}.
\end{equation}
Observe that when $\boldsymbol{\Omega}_i = \boldsymbol{\Omega}$, scalar products remain the same
\begin{equation}
	\boldsymbol{\mu}_{i}^{\top} \boldsymbol{\mu}_{j} = \boldsymbol{\nu}_{i}^{\top}\exp\left(-\boldsymbol{\Omega} t\right)\exp\left(\boldsymbol{\Omega} t\right) \boldsymbol{\nu}_{j} = \boldsymbol{\nu}_{i}^{\top} \boldsymbol{\nu}_{j},
\end{equation}
which implies $E_{K}(\boldsymbol{\mu}_{i}) = E_{K}(\boldsymbol{\nu}_{i})$ when $\boldsymbol{\Omega}_i = \boldsymbol{\Omega}$ since by assumption it depends only on the scalar products.

To rewrite equations of motion in new variables we observe that
\begin{equation}
	\frac{\partial E_{K}}{\partial \left(\boldsymbol{\mu}_{i}\right)_{\alpha}} = \sum_{\beta} \frac{\partial E_{K}}{\partial \left(\boldsymbol{\nu}_{i}\right)_{\beta}}\frac{\partial \left(\boldsymbol{\nu}_{i}\right)_{\beta}}{\partial \left(\boldsymbol{\mu}_{i}\right)_{\alpha}} = \sum_{\beta} \frac{\partial E_{K}}{\partial \left(\boldsymbol{\nu}_{i}\right)_{\beta}} \left(\exp\left(-\boldsymbol{\Omega}_i t\right)\right)_{\beta\alpha} = \sum_{\beta} \frac{\partial E_{K}}{\partial \left(\boldsymbol{\nu}_{i}\right)_{\beta}} \left(\exp\left(\boldsymbol{\Omega}_i t\right)\right)_{\alpha\beta},
\end{equation}
or in matrix form
\begin{equation}
	\frac{\partial E_{K}}{\partial \boldsymbol{\mu}_{i}} = \exp\left(\boldsymbol{\Omega}_i t\right) \frac{\partial E_{K}}{\partial \boldsymbol{\nu}_{i}}.
\end{equation}
For time derivative we find
\begin{equation}
	\dot{\boldsymbol{\mu}}_i = \boldsymbol{\Omega}_i\exp\left(\boldsymbol{\Omega}_i t\right) \boldsymbol{\nu}_{i} + \exp\left(\boldsymbol{\Omega}_i t\right) \dot{\boldsymbol{\nu}}_{i} = \boldsymbol{\Omega}_i\boldsymbol{\mu}_i + \exp\left(\boldsymbol{\Omega}_i t\right) \dot{\boldsymbol{\nu}}_{i}.
\end{equation}
With that we can rewrite dynamical system~(\ref{eq:Kuramoto_memory}) in new variables
\begin{equation}
	\exp\left(\boldsymbol{\Omega}_i t\right) \dot{\boldsymbol{\nu}}_{i} = -\left(\boldsymbol{I} - \exp\left(\boldsymbol{\Omega}_i t\right)\boldsymbol{\nu}_i \boldsymbol{\nu}_i^{\top} \exp\left(-\boldsymbol{\Omega}_i t\right) \right) \exp\left(\boldsymbol{\Omega}_i t\right) \frac{\partial E_{K}}{\partial \boldsymbol{\nu}_{i}} \Rightarrow \dot{\boldsymbol{\nu}}_{i} =  -\left(\boldsymbol{I} - \boldsymbol{\nu}_i \boldsymbol{\nu}_i^{\top}\right) \frac{\partial E_{K}}{\partial \boldsymbol{\nu}_{i}}.
\end{equation}
So we see that assumption $\boldsymbol{\Omega}_i = \boldsymbol{\Omega}$ reduces Kuramoto memory model to the constrained gradient flow.

Since energy function has the same form in new variables we find
\begin{equation}
	\dot{E}_{K}(\boldsymbol{\nu}) = \sum_{i}\left(\frac{\partial E_{K}(\boldsymbol{\nu})}{\partial \boldsymbol{\nu}_{i}}\right)^{\top} \dot{\boldsymbol{\nu}}_{i} = -\sum_{i} \left(\frac{\partial E_{K}(\boldsymbol{\nu})}{\partial \boldsymbol{\nu}_{i}}\right)^{\top} \left(\boldsymbol{I} - \boldsymbol{\nu}_i \boldsymbol{\nu}_i^{\top}\right) \frac{\partial E_{K}(\boldsymbol{\nu})}{\partial \boldsymbol{\nu}_{i}}.
\end{equation}
Using identity for derivatives and relation between $\boldsymbol{\nu}_i$ and $\boldsymbol{\mu}_i$ we change variables to the original ones
\begin{equation}
	\dot{E}_{K}(\boldsymbol{\mu}) = -\sum_{i} \left(\frac{\partial E_{K}(\boldsymbol{\mu})}{\partial \boldsymbol{\mu}_{i}}\right)^{\top} \left(\boldsymbol{I} - \boldsymbol{\mu}_i \boldsymbol{\mu}_i^{\top}\right) \frac{\partial E_{K}(\boldsymbol{\mu})}{\partial \boldsymbol{\mu}_{i}}.
\end{equation}
Since $\boldsymbol{\mu}_i^{\top}\boldsymbol{\mu}_i = 1$, matrix $\boldsymbol{I} - \boldsymbol{\mu}_i \boldsymbol{\mu}_i^{\top}$ is orthogonal projector which is positive semidefinite matrix. From that we conclude
\begin{equation}
	\dot{E}_{K}(\boldsymbol{\mu}) = -\sum_{i} \left(\frac{\partial E_{K}(\boldsymbol{\mu})}{\partial \boldsymbol{\mu}_{i}}\right)^{\top} \left(\boldsymbol{I} - \boldsymbol{\mu}_i \boldsymbol{\mu}_i^{\top}\right) \frac{\partial E_{K}(\boldsymbol{\mu})}{\partial \boldsymbol{\mu}_{i}}\leq 0.
\end{equation}
\subsection{Direct proof}
We define $s_{\alpha\beta} = \boldsymbol{\mu}_{\alpha}^\top\boldsymbol{\mu}_{\beta}$. Using that energy is function of $s_{\alpha\beta}$ only we find
\begin{equation}
	\frac{\partial E_{K}(\boldsymbol{\mu}) }{\partial \boldsymbol{\mu}_i} = \frac{1}{2}\sum_{\alpha\beta} \frac{\partial E_{K} }{\partial s_{\alpha\beta}}\frac{\partial s_{\alpha\beta}}{\partial \boldsymbol{\mu}_i} = \frac{1}{2}\sum_{\alpha\beta} \frac{\partial E_{K} }{\partial s_{\alpha\beta}}\left(\boldsymbol{\mu}_{\beta}\delta_{i\alpha} + \boldsymbol{\mu}_{\alpha}\delta_{i\beta}\right) = \sum_{\beta}\frac{\partial E_{K}}{\partial s_{i\beta}}\boldsymbol{\mu}_{\beta}.
\end{equation}
This derivative appears in the equations of motion and also in the time derivative of the energy function. With this identity we find
\begin{equation}
	\dot{E}_{K}(\boldsymbol{\mu}) = \sum_{i}\left(\frac{\partial E_{K}(\boldsymbol{\mu})}{\partial \boldsymbol{\mu}_{i}}\right)^{\top} \dot{\boldsymbol{\mu}}_{i} = \sum_{ij} \frac{\partial E_{K}}{\partial s_{ij}}\boldsymbol{\mu}_{j}^{\top}\dot{\boldsymbol{\mu}}_{i} = \sum_{ij} \frac{\partial E_{K}}{\partial s_{ij}} \boldsymbol{\mu}_{j}^{\top}\left(\boldsymbol{\Omega}_{i}\boldsymbol{\mu}_{i} -  \left(\boldsymbol{I} - \boldsymbol{\mu}_i \boldsymbol{\mu}_i^{\top}\right)\sum_{\beta}\frac{\partial E_{K}}{\partial s_{i\beta}}\boldsymbol{\mu}_{\beta}\right).
\end{equation}
The first term with $\boldsymbol{\Omega}_i$ is zero when $\boldsymbol{\Omega}_i = \boldsymbol{\Omega}$ since  $s_{ij} = s_{ji}$ but $\boldsymbol{\mu}_{j}^\top \boldsymbol{\Omega}\boldsymbol{\mu}_{i} = -\boldsymbol{\mu}_{i}^\top \boldsymbol{\Omega}\boldsymbol{\mu}_{j}$. The remaining term gives
\begin{equation}
	\dot{E}_{K}(\boldsymbol{\mu}) = -\sum_{i}\left(\sum_{j}\frac{\partial E_{K}}{\partial s_{ij}} \boldsymbol{\mu}_{j}\right)^{\top}\left(\boldsymbol{I} - \boldsymbol{\mu}_i \boldsymbol{\mu}_i^{\top}\right)\left(\sum_{\beta}\frac{\partial E_{K}}{\partial s_{i\beta}}\boldsymbol{\mu}_{\beta}\right) = -\sum_{i} \left(\frac{\partial E_{K}(\boldsymbol{\mu})}{\partial \boldsymbol{\mu}_{i}}\right)^{\top} \left(\boldsymbol{I} - \boldsymbol{\mu}_i \boldsymbol{\mu}_i^{\top}\right) \frac{\partial E_{K}(\boldsymbol{\mu})}{\partial \boldsymbol{\mu}_{i}}\leq 0.
\end{equation}
\section{Encoding information in scalar products}
\label{appendix:scalar_products_encoding}
As explained in the main text, under conditions of Theorem~\ref{th:Lyapunov_Kuramoto}, Kuramoto model reduces to gradient flow with energy functions that depend only on scalar products. Initial conditions for oscillatory neurons select a particular point on the energy landscape and starting from this point the dynamical system evolves toward the state corresponding to the local energy minimum. It is natural then to encode all information in the scalar products.

Suppose we need to encode $N-1$ real numbers $v_i$ from the interval $[-1, 1]$. Evidently, we can select
\begin{equation}
	\boldsymbol{\mu}_{1} = 
	\begin{pmatrix}
	1\\
	0\\
	\vdots\\
	0	
	\end{pmatrix},
	\boldsymbol{\mu}_{2} = 
	\begin{pmatrix}
	v_1\\
	1 - v_1^2\\
	0\\
	\vdots\\
	0	
	\end{pmatrix},\dots,
	\boldsymbol{\mu}_{N} = 
	\begin{pmatrix}
	v_{N-1}\\
	1 - v_{N-1}^2\\
	0\\
	\vdots\\
	0	
	\end{pmatrix},
\end{equation}
so vectors are normalised $\boldsymbol{\mu}_i^{\top}\boldsymbol{\mu}_{i} = 1$ and scalar products encode required information $\boldsymbol{\mu}_1^{\top}\boldsymbol{\mu}_{i+1} = v_{i}$.

This encoding scheme is valid, but one may suspect it is not very efficient since it does not make use of other components of vectors $v_{i}$ besides the first two. We will argue here that in the worst case this intuition is wrong and $N$ oscillatory units can encode only $N-1$ real numbers from $[-1, 1]$ regardless of the chosen scheme and dimension of vectors $v_{i}$.

To see that consider $\boldsymbol{v} = \boldsymbol{1}$.

First observe that if scalar product of two vectors on the sphere $\boldsymbol{\mu}_{i}^\top \boldsymbol{\mu}_{i} = \boldsymbol{\mu}_{k}^\top \boldsymbol{\mu}_{k} = 1$ equals one $\boldsymbol{\mu}_{k}^\top \boldsymbol{\mu}_{i} = 1$, these vectors  coincide $\boldsymbol{\mu}_{k} = \boldsymbol{\mu}_{i}$. This follows from Cauchy–Schwarz inequality.

This implies if we select two vectors $\boldsymbol{\mu}_{i}$, $\boldsymbol{\mu}_{j}$ and encode $1$ in their scalar product, we lose $D$ degrees of freedom.

Any coding strategy that encode components of $\boldsymbol{v}$ into scalar products is defined by the order in which pair of vectors is selected $\boldsymbol{\mu}_{i_1}^{\top}\boldsymbol{\mu}_{j_1}^{\top} = v_{1}$, $\boldsymbol{\mu}_{i_2}^{\top}\boldsymbol{\mu}_{j_2}^{\top} = v_{2}$, $\dots$. The total number of degrees of freedom that we have equals $(N-1)\times D$ not $N\times D$, because scalar products and dynamical systems are invariant under global rotation. Since each time we pick a novel pair $i_k\neq j_k$ we lose $D$ degrees of freedom, we can encode at most $N-1$ components.

The worst case result is somewhat frustrating, but we do not expect it to happen too often especially if natural data is considered. Given that, it is an interesting question what happens with ``encoding capacity'' in the typical case.
\section{Lyapunov function for Hopfield-Kuramoto associative memory}
\label{appendix:Hopfield_Kuramoto}
First, we split energy function~(\ref{eq:Hopfield_Kuramoto_energy}) of joint Hopfield-Kuramoto model on three terms
\begin{equation}
	\begin{split}
	&E_{HK}\left(\boldsymbol{x}, \boldsymbol{\mu}\right) = \kappa_{H}E_{H}\left(\boldsymbol{x}\right) + \kappa_{K} E_{K}\left(\boldsymbol{\mu}\right) + \Delta E_{HK}\left(\boldsymbol{x}, \boldsymbol{\mu}\right),\\
	&\Delta E_{HK}\left(\boldsymbol{x}, \boldsymbol{\mu}\right) =  - \frac{1}{2}\sum_{ij}G_{ij}\left(g(\boldsymbol{x})_{i} g(\boldsymbol{x})_{j}\right)\chi_{ij}(\boldsymbol{\mu}_{i}^{\top}\boldsymbol{\mu}_{j}).
	\end{split}
\end{equation}
Time derivative of energy function reads
\begin{equation}
	\begin{split}
		&\dot{E}_{HK}\left(\boldsymbol{x}, \boldsymbol{\mu}\right) = \sum_{i}\frac{\partial E_{HK}\left(\boldsymbol{x}, \boldsymbol{\mu}\right)}{\partial x_i}\dot{x}_i + \sum_{i}\left(\frac{\partial E_{HK}\left(\boldsymbol{x}, \boldsymbol{\mu}\right)}{\partial \boldsymbol{\mu}_i}\right)^{\top}\dot{\mu}_i\\
		&=\kappa_{H}\sum_{i}\underbrace{\frac{\partial E_{H}\left(\boldsymbol{x}\right)}{\partial x_i}}_{\text{Appendix \ref{appendix:Lyapunov_Hopfield}}}\dot{x}_i +  \kappa_{K}\sum_{i}\underbrace{\left(\frac{\partial E_{K}\left(\boldsymbol{\mu}\right)}{\partial \boldsymbol{\mu}_i}\right)^{\top}}_{\text{Appendix \ref{appendix:Lyapunov_Kuramoto}}}\dot{\mu}_i + \sum_{i}\frac{\partial \Delta E_{HK}\left(\boldsymbol{x}, \boldsymbol{\mu}\right)}{\partial x_i}\dot{x}_i + \sum_{i}\left(\frac{\partial \Delta E_{HK}\left(\boldsymbol{x}, \boldsymbol{\mu}\right)}{\partial \boldsymbol{\mu}_i}\right)^{\top}\dot{\mu}_i.
	\end{split}
\end{equation}
Since joint model significantly reuses equations of motion of uncoupled models, we can borrow most of the results from Appendix~\ref{appendix:Lyapunov_Hopfield} and Appendix~\ref{appendix:Lyapunov_Kuramoto}. The last two terms in the equation above are new, so we focus on them. Hopfield and Kuramoto parts can be considered separately.

For the Hopfield model we obtain
\begin{multline}
	\frac{\partial \Delta E_{HK}\left(\boldsymbol{x}, \boldsymbol{\mu}\right)}{\partial x_i} = - \frac{1}{2}\sum_{\alpha\beta}G^{'}_{\alpha\beta}\left(g(\boldsymbol{x})_{\alpha} g(\boldsymbol{x})_{\beta}\right)\chi_{\alpha\beta}(\boldsymbol{\mu}_{\alpha}^{\top}\boldsymbol{\mu}_{\beta}) \frac{\partial \left(g(\boldsymbol{x})_{\alpha} g(\boldsymbol{x})_{\beta}\right)}{\partial x_i} \\
	= - \frac{1}{2}\sum_{\alpha\beta}G^{'}_{\alpha\beta}\left(g(\boldsymbol{x})_{\alpha} g(\boldsymbol{x})_{\beta}\right)\chi_{\alpha\beta}(\boldsymbol{\mu}_{\alpha}^{\top}\boldsymbol{\mu}_{\beta})\left(\Lambda_{i\alpha}g(\boldsymbol{x})_{\beta} + \Lambda_{i\beta}g(\boldsymbol{x})_{\alpha}\right) \\
	= - \sum_{\alpha\beta}\Lambda_{i\alpha} G^{'}_{\alpha\beta}\left(g(\boldsymbol{x})_{\alpha} g(\boldsymbol{x})_{\beta}\right)\chi_{\alpha\beta}(\boldsymbol{\mu}_{\alpha}^{\top}\boldsymbol{\mu}_{\beta})g(\boldsymbol{x})_{\beta} = - \sum_{\alpha}\Lambda_{i\alpha}\Delta x_{\alpha}\left(\boldsymbol{x}, \boldsymbol{\mu}\right).
\end{multline}
Reusing results from Appendix~\ref{appendix:Lyapunov_Hopfield} we find
\begin{equation}
	\kappa_{H}\sum_{i}\frac{\partial E_{H}\left(\boldsymbol{x}\right)}{\partial x_i}\dot{x}_i + \sum_{i}\frac{\partial \Delta E_{HK}\left(\boldsymbol{x}, \boldsymbol{\mu}\right)}{\partial x_i}\dot{x}_i = -\kappa_{H}\sum_{ij}\Lambda_{ij}\left(\sum_{k}W_{jk}g(\boldsymbol{x})_{k} - x_{j} + b_{j} + \frac{1}{\kappa_{H}}\Delta x_{j}\left(\boldsymbol{x}, \boldsymbol{\mu}\right)\right)\dot{x}_i.
\end{equation}
It is easy to see from equations of motion~(\ref{eq:Hopfield_Kuramoto_memory}) that term in brackets is $\dot{x}_{i}$ so we reproduce the first term in derivative of energy function~(\ref{eq:Hopfield_Kuramoto_energy_derivative}).

For the Kuramoto model we find
\begin{multline}
	\frac{\partial \Delta E_{HK}\left(\boldsymbol{x}, \boldsymbol{\mu}\right)}{\partial \boldsymbol{\mu}_i} =  - \frac{1}{2}\sum_{\alpha\beta}G_{\alpha\beta}\left(g(\boldsymbol{x})_{\alpha} g(\boldsymbol{x})_{\beta}\right)\chi^{'}_{\alpha\beta}(\boldsymbol{\mu}_{\alpha}^{\top}\boldsymbol{\mu}_{\beta}) \frac{\partial \left(\boldsymbol{\mu}_{\alpha}^{\top}\boldsymbol{\mu}_{\beta}\right)}{\partial \boldsymbol{\mu}_i} \\
	= - \frac{1}{2}\sum_{\alpha\beta}G_{\alpha\beta}\left(g(\boldsymbol{x})_{\alpha} g(\boldsymbol{x})_{\beta}\right)\chi^{'}_{\alpha\beta}(\boldsymbol{\mu}_{\alpha}^{\top}\boldsymbol{\mu}_{\beta})\left(\delta_{\alpha i}\boldsymbol{\mu}_{\beta} + \delta_{\beta i}\boldsymbol{\mu}_{\alpha}\right)\\
	= -\sum_{\beta}G_{i\beta}\left(g(\boldsymbol{x})_{i} g(\boldsymbol{x})_{\beta}\right)\chi^{'}_{i\beta}(\boldsymbol{\mu}_{i}^{\top}\boldsymbol{\mu}_{\beta})\boldsymbol{\mu}_{\beta} = -\delta \boldsymbol{\mu}_{i}(\boldsymbol{x}, \boldsymbol{\mu})
\end{multline}
This result allows us to rewrite equations of motion~(\ref{eq:Hopfield_Kuramoto_memory}) for oscillatory neurons in the following form
\begin{multline}
	\dot{\boldsymbol{\mu}}_{i} = \boldsymbol{\Omega}_{i}\boldsymbol{\mu}_i + \left(\boldsymbol{I} - \boldsymbol{\mu}_i\boldsymbol{\mu}_i^{\top}\right)\left(-\frac{\partial E_{K}}{\partial \boldsymbol{\mu}_i} + \frac{1}{\kappa_{K}}\delta \boldsymbol{\mu}_{i}(\boldsymbol{x}, \boldsymbol{\mu})\right) \\= \boldsymbol{\Omega}_{i}\boldsymbol{\mu}_i - \frac{1}{\kappa_{K}}\left(\boldsymbol{I} - \boldsymbol{\mu}_i\boldsymbol{\mu}_i^{\top}\right)\left(\kappa_{K}\frac{\partial E_{K}}{\partial \boldsymbol{\mu}_i} + \frac{\partial \Delta E_{HK}\left(\boldsymbol{x}, \boldsymbol{\mu}\right)}{\partial \boldsymbol{\mu}_i}\right) = \boldsymbol{\Omega}_{i}\boldsymbol{\mu}_i - \frac{1}{\kappa_{K}}\left(\boldsymbol{I} - \boldsymbol{\mu}_i\boldsymbol{\mu}_i^{\top}\right) \frac{\partial E_{HK}\left(\boldsymbol{x}, \boldsymbol{\mu}\right)}{\partial \boldsymbol{\mu}_i}.
\end{multline}
So we can simplify expression for time derivative of energy
\begin{equation}
	\sum_{i}\left(\frac{\partial E_{HK}\left(\boldsymbol{x}, \boldsymbol{\mu}\right)}{\partial \boldsymbol{\mu}_i}\right)^{\top}\dot{\mu}_i = - \frac{1}{\kappa_{K}}\sum_{i}\left(\frac{\partial E_{HK}\left(\boldsymbol{x}, \boldsymbol{\mu}\right)}{\partial \boldsymbol{\mu}_i}\right)^{\top}\left(\boldsymbol{I} - \boldsymbol{\mu}_i\boldsymbol{\mu}_i^{\top}\right) \frac{\partial E_{HK}\left(\boldsymbol{x}, \boldsymbol{\mu}\right)}{\partial \boldsymbol{\mu}_i} + \sum_{i}\left(\frac{\partial E_{HK}\left(\boldsymbol{x}, \boldsymbol{\mu}\right)}{\partial \boldsymbol{\mu}_i}\right)^{\top}\boldsymbol{\Omega}_{i}\boldsymbol{\mu}_i.
\end{equation}
The first term in the equation above reproduces the second term in the derivative of energy function~(\ref{eq:Hopfield_Kuramoto_energy_derivative}). Now, we need to prove that the second term is zero. For that we use result from Appendix~\ref{appendix:Lyapunov_Kuramoto} that gives us
\begin{equation}
	\sum_{i}\left(\frac{\partial E_{HK}\left(\boldsymbol{x}, \boldsymbol{\mu}\right)}{\partial \boldsymbol{\mu}_i}\right)^{\top}\boldsymbol{\Omega}_{i}\boldsymbol{\mu}_i = \sum_{i\beta} \left(\kappa_{K}\frac{\partial E_{K}}{\partial s_{i\beta}} - G_{i\beta}\left(g(\boldsymbol{x})_{i} g(\boldsymbol{x})_{\beta}\right)\chi^{'}_{i\beta}(\boldsymbol{\mu}_{i}^{\top}\boldsymbol{\mu}_{\beta})\right)\boldsymbol{\mu}_\beta^{\top}\boldsymbol{\Omega}_{i}\boldsymbol{\mu}_i. 
\end{equation}
Expression in brackets is symmetric with respect to $i$ and $\beta$ and if $\boldsymbol{\Omega}_{i}= \boldsymbol{\Omega}$ scalar product is skew-symmetric $\boldsymbol{\mu}_\beta^{\top}\boldsymbol{\Omega}\boldsymbol{\mu}_i = -\boldsymbol{\mu}_i^{\top}\boldsymbol{\Omega}\boldsymbol{\mu}_\beta$, so contraction is zero.

So we proved that the derivative of energy is given by the expression~(\ref{eq:Hopfield_Kuramoto_energy_derivative}) as claimed in the theorem. Derivative is non-increasing on trajectories since individual terms are quadratic forms with negative semidefinite matrices

\section{Hierarchical Hopfield-Kuramoto models}
\label{appendix:Models}
In Section~\ref{section:Hopfield-Kuramoto} we showed a general construction of the joint Hopfield-Kuramoto associative memory model. Here we will provide several deep architectures with interacting oscillatory neurons and threshold units. All models from this section have global energy function non-increasing on trajectories of joint Hopfield-Kuramoto model~(\ref{eq:Hopfield_Kuramoto_memory}). We do not provide separate proofs for each case since all of them can be reduced to Theorem~\ref{th:Lyapunov_Hopfield_Kuramoto} with particular choice of weights, activations and Kuramoto energy function $E_{K}$.

In all examples below lower index indicate layer. To present results in a more compact way we denote $\left(\boldsymbol{P}_{\boldsymbol{\mu}}\right)_{ij} = (\boldsymbol{I} - \boldsymbol{\mu}_{i}\boldsymbol{\mu}_{i}^{\top})I_{ij}$ and stack $i=1,\dots,N_k$ oscilatory neurons $\boldsymbol{\mu}_{i}\in\mathbb{R}^{D+1}$ in a single vector $\boldsymbol{\mu}^{\top} = \begin{pmatrix}\boldsymbol{\mu}^{\top}_1 & \dots & \boldsymbol{\mu}^{\top}_{N_K}\end{pmatrix} \in \mathbb{R}^{N\times(D+1)}$. So with slight abuse of notation we now use $\boldsymbol{\mu}_{i}$ to indicate the state of \textit{all} oscillatory neurons rather than the state of $i$-th oscillatory neuron as in Kuramoto associative memory~(\ref{eq:Kuramoto_memory}). In this notation equation~(\ref{eq:Kuramoto_memory}) with $\boldsymbol{\Omega}_i = \boldsymbol{\Omega}$ becomes $\dot{\boldsymbol{\mu}} = \boldsymbol{I}\otimes \boldsymbol{\Omega} \boldsymbol{\mu} - \boldsymbol{P}_{\mu}\frac{\partial E_{K}}{\partial \boldsymbol{\mu}}$.

In all cases we consider a network with $K$ combined threshold-oscillatory neurons. The state of the layer is a pair $(\boldsymbol{x}_{i} \in \mathbb{R}^{N_i}, \boldsymbol{\mu}_{i}\in \mathbb{R}^{N_i\times (D+1)})$.

\subsection{Networks with oscillatory lateral connections}
The most natural way to build hierarchical networks is to use threshold units for deep connections and oscillatory neurons for lateral connections. We will provide three examples that differ by activation function and structure of the linear layer. In all examples we use minimal coupling between oscillatory and threshold units.
\subsubsection{Fully connected, ReLU activation}
Equations of motion for layer $i=2,\dots,K-1$ are
\begin{equation}
	\label{eq:MLP_ReLU_1}
	\begin{aligned}
		\dot{\boldsymbol{x}}_i &= \sum_{p=\pm1}\boldsymbol{W}_{i i+p} {\sf ReLU}(\boldsymbol{x}_{i+p}) - \boldsymbol{x}_i + \boldsymbol{b}_{i}  + \frac{1}{\kappa_{H}} \boldsymbol{S}_{i} \odot \boldsymbol{\mathcal{G}}(\boldsymbol{\mu}_i, \boldsymbol{\mu}_i) {\sf ReLU}(\boldsymbol{x}_{i}),\\
		\dot{\boldsymbol{\mu}}_i &= \boldsymbol{I}\otimes \boldsymbol{\Omega} \boldsymbol{\mu}_i - \boldsymbol{P}_{\mu_i}\frac{\partial E_{K}}{\partial \boldsymbol{\mu}} + \boldsymbol{P}_{\mu_i}\frac{1}{\kappa_{K}} \left({\sf ReLU}(\boldsymbol{x}_{i}){\sf ReLU}(\boldsymbol{x}_{i})^\top \odot \boldsymbol{S}_{i}\right) \otimes \boldsymbol{I} \boldsymbol{\mu}_i.
	\end{aligned}
\end{equation}
where $\boldsymbol{W}_{ij}^\top = \boldsymbol{W}_{ij}$ are weights for feedforward and feedback connections, $\boldsymbol{S}_{i} = \boldsymbol{S}_{i}^{\top}$ are weights of lateral connections. For layer $1$ and $K$ one needs to omit terms with $\boldsymbol{x}_{-1}$ and $\boldsymbol{x}_{K+1}$ from the equation for threshold units.

Energy function reads
\begin{multline}
	E_{HK} = \sum_{i=1}^{K}\left(\boldsymbol{x}_i - \boldsymbol{b}_i\right)^{\top}{\sf ReLU}(\boldsymbol{x}_i) - \frac{1}{2}\sum_{i=1}^{K}\boldsymbol{1}^{\top}{\sf ReLU}^2(\boldsymbol{x}_i) - \sum_{i=1}^{K-1}{\sf ReLU}(\boldsymbol{x}_{i})^\top \boldsymbol{W}_{ii+1}{\sf ReLU}(\boldsymbol{x}_{i+1}) + E_{K} \\+ \frac{1}{2}\sum_{i=1}^{K}\left\|{\sf ReLU}(\boldsymbol{x}_{i}){\sf ReLU}(\boldsymbol{x}_{i})^\top \odot \boldsymbol{S}_{i} \odot \boldsymbol{\mathcal{G}}(\boldsymbol{\mu}_i, \boldsymbol{\mu}_i)\right\|_{F}^2.
\end{multline}

From equations of motion we see that layers $i$ and $i+1$ are coupled through interaction of threshold units. Threshold units interact with oscillatory units only within each layer, i.e., they form lateral connections. There is no direct interaction between oscillatory units of distinct layers as long as $E_K$ does not contain them.

\subsubsection{Convolutional, ReLU activation}
It is possible to replace dense layers with convolutional layers for both deep and lateral connections. Since the convolution layer is merely a structured linear layer, one only needs to slightly change interpretation and notation. For example, if network process images, within each layer in~(\ref{eq:MLP_ReLU_1}) index of neuron is triple $(h, w, c)$ -- heights, width and channel, so the state of threshold-oscillatory unit is a pair $(\boldsymbol{x}_{i} \in \mathbb{R}^{H_i\times W_i\times C_i}, \boldsymbol{\mu}_{i}\in \mathbb{R}^{H_i\times W_i\times C_i\times (D+1)})$. Next, convolution operation for both threshold and oscillatory unit has the same form
\begin{equation}
	x_{h,w,c} = \sum_{i=-k}^{k}\sum_{j=-k}^{k}\sum_{o=1}^{N_c} w_{i,j,o,c}y_{h+i,w+j,o},
\end{equation}
but for oscillatory units $x_{h,w,c}$ and $y_{h,w,c}$ are replaced by vectors $\boldsymbol{\mu}_{h,w,c}$, $\boldsymbol{\nu}_{h,w,c}$ since each oscillatory units has $D$ rotational degrees of freedom. Other operations are also straightforward. For example, Gram matrix will become
\begin{equation}
	\left(\boldsymbol{\mathcal{G}}(\boldsymbol{\mu}, \boldsymbol{\mu})\right)_{\left(h_1,w_1,c_1\right)\left(h_2,w_2,c_2\right)} = \boldsymbol{\mu}_{h_1,w_1,c_1}^{\top} \boldsymbol{\mu}_{h_2,w_2,c_2},
\end{equation}
and the gating mechanism will make parameters of the lateral convolution kernel given by matrix $\boldsymbol{S}$ explicitly space-dependent.

\subsubsection{Fully connected, softmax}
Dynamical system~(\ref{eq:MLP_ReLU_1}) will have Lyapunov function if we replace ${\sf ReLU}(\boldsymbol{x})$ by softmax function ${\sf sm}(\boldsymbol{x}) = \exp(\boldsymbol{x}) / \left(\boldsymbol{1}^\top\exp(\boldsymbol{x})\right)$. More specifically, Lyapunov function becomes
\begin{multline}
	E_{HK} = \sum_{i=1}^{K}\left(\boldsymbol{x}_i - \boldsymbol{b}_i\right)^{\top}{\sf sm}(\boldsymbol{x}_i) - \sum_{i=1}^{K}\log\left(\boldsymbol{1}^{\top}\exp(\boldsymbol{x}_i)\right) - \sum_{i=1}^{K-1}{\sf sm}(\boldsymbol{x}_{i})^\top \boldsymbol{W}_{ii+1}{\sf sm}(\boldsymbol{x}_{i+1}) + E_{K} \\+ \frac{1}{2}\sum_{i}\left\|{\sf sm}(\boldsymbol{x}_{i}){\sf sm}(\boldsymbol{x}_{i})^\top \odot \boldsymbol{S}_{i} \odot \boldsymbol{\mathcal{G}}(\boldsymbol{\mu}_i, \boldsymbol{\mu}_i)\right\|_{F}^2.
\end{multline}
It is also possible to implement attention mechanisms from \cite{ramsauer2020hopfield} as explained in \cite{krotov2020large}. It is also possible to use many other activation functions \cite{krotov2021hierarchical}.

\subsection{Networks with threshold lateral connections}
Oscillatory units can also be used to form deep connections. Here we provide an example with ${\sf ReLU}$ activations and fully-connected layers, but other activations can be used as well. Equations of motion for layer $i=2,\dots,K-1$ are
\begin{equation}
	\label{eq:MLP_ReLU_2}
	\begin{aligned}
		\dot{\boldsymbol{x}}_i &= \boldsymbol{W}_{i} {\sf ReLU}(\boldsymbol{x}_{i}) - \boldsymbol{x}_i + \boldsymbol{b}_{i} + \frac{1}{\kappa_{H}} \boldsymbol{S}_{i} \odot \boldsymbol{\mathcal{G}}(\boldsymbol{\mu}_i, \boldsymbol{\mu}_i) {\sf ReLU}(\boldsymbol{x}_{i}),\\
		\dot{\boldsymbol{\mu}}_i &= \boldsymbol{I}\otimes \boldsymbol{\Omega} \boldsymbol{\mu}_i - \sum_{p=\pm1}\boldsymbol{P}_{\mu_i}\left(\boldsymbol{W}_{ii+p}\odot{\sf ReLU}\left(\boldsymbol{\mathcal{G}}(\boldsymbol{\mu}_i, \boldsymbol{\mu}_{i+p})\right) \right)\otimes\boldsymbol{I}\boldsymbol{\mu}_{i+p} + \boldsymbol{P}_{\mu_i}\frac{1}{\kappa_{K}} \left({\sf ReLU}(\boldsymbol{x}_{i}){\sf ReLU}(\boldsymbol{x}_{i})^\top \odot \boldsymbol{S}_{i}\right) \otimes \boldsymbol{I} \boldsymbol{\mu}_i.
	\end{aligned}
\end{equation}
where $\boldsymbol{W}_{ij}^\top = \boldsymbol{W}_{ij}$ are weights for feedforward and feedback connections, $\boldsymbol{S}_{i} = \boldsymbol{S}_{i}^{\top}$ and $\boldsymbol{W}_{i} = \boldsymbol{W}_{i}^{\top}$ are weights of lateral connections. For layer $1$ and $K$ one needs to omit terms with $\boldsymbol{\mu}_{-1}$ and $\boldsymbol{\mu}_{K+1}$ from the equation for oscillatory units.

Energy function reads
\begin{multline}
	E_{HK} = \sum_{i=1}^{K}\left(\boldsymbol{x}_i - \boldsymbol{b}_i\right)^{\top}{\sf ReLU}(\boldsymbol{x}_i) - \frac{1}{2}\sum_{i=1}^{K}\boldsymbol{1}^{\top}{\sf ReLU}^2(\boldsymbol{x}_i) - \frac{1}{2}\sum_{i=1}^{K}{\sf ReLU}(\boldsymbol{x}_{i})^\top \boldsymbol{W}_{i}{\sf ReLU}(\boldsymbol{x}_{i})\\ + \frac{1}{2}\sum_{i=1}^{K}\left\|{\sf ReLU}(\boldsymbol{x}_{i}){\sf ReLU}(\boldsymbol{x}_{i})^\top \odot \boldsymbol{S}_{i} \odot \boldsymbol{\mathcal{G}}(\boldsymbol{\mu}_i, \boldsymbol{\mu}_i)\right\|_{F}^2 + \frac{1}{4} \sum_{i=1}^{K-1} \left\|\boldsymbol{W}_{ii+1}\odot{\sf ReLU}^2\left(\boldsymbol{\mathcal{G}}(\boldsymbol{\mu}_i, \boldsymbol{\mu}_{i+1})\right)\right\|_{F}^2
\end{multline}

Wee see that this time oscillatory neurons of layers $i$ and $i+1$ interact directly, whereas threshold units interact only within individual layers. Interaction of oscillatory units contain self-gating, so deep connection looks less natural than the ones formed by threshold units.

Dynamical system~(\ref{eq:MLP_ReLU_2}) can be extended on other activation functions and it is possible to use convolutional layers the same way as explained in the previous section.

\section{Description of fine-tuning experiments}
\label{appendix:Exp_description}
The code that reproduces our numerical results is available in \url{https://github.com/vlsf/HKmemory}. Trained models are available in \url{https://disk.yandex.ru/d/n5eQAPOSu4D4pw}.

To implement training of memory models we used JAX \cite{jax2018github}, Optax \cite{deepmind2020jax}, Equinox \cite{kidger2021equinox}, Diffrax \cite{kidger2021on}. NetworkX was used to generate Watts-Strogatz small world NetworkX \cite{SciPyProceedings_11}. For convenience we split our description into several parts.
\subsection{ODE solver}
ODE is solved with Tsitouras' 5/4 method, step size is adjusted with PID controller, integration interval is $[0, 1]$, initial time step $10^{-2}$, maximal number of steps allowed is $5000$, library used is \cite{kidger2021on}.
\subsection{Optimisation}
In all cases we use Lion optimizer \cite{chen2023symbolic}, with learning rate $10^{-4}$, exponential learning rate decay with rate $\gamma=0.5$ and $1000$ transition steps, number of weight updates is $9000$, batch size is $100$, number of train examples is $60000$, number of test examples is $10000$, library used in Optax \cite{deepmind2020jax}.
\subsection{Hopfield subnetwork}
Model is described in the equation~(\ref{eq:exp_model}), input is augmented with $300$ additional elements, MNIST images were normalised to $[0, 1]$ range, weight matrix is dense, the total number of parameters of the model is $\simeq 12\times 10^{5}$, classes were encoded as $2\boldsymbol{e}_i - 1$ where $\boldsymbol{e}_i\in\mathbb{R}^{10}$ is $i$-th column of identity matrix, training loss is $L_2$, and predicted classes are last $10$ components of $\boldsymbol{x}(1)$ vector, library used in Equinox \cite{kidger2021equinox}.
\subsection{Kuramoto subnetwork}
Model is described in the equation~(\ref{eq:exp_model}), input is augmented with $300$ additional elements, MNIST images were normalised to $[0, 1]$ range, weight matrices are sparse, sparsity patter in generated with Watts-Strogatz small world as implemented in NetworkX with $p=0.1$ and $k = 150, 250, 350, 450$, rounded number of parameters for Kuramoto subnetworks are given in Table~\ref{table:N_params_Kuramoto}, to encode MNIST digit into initial conditions we set $\boldsymbol{\mu}_{1} = \boldsymbol{e}_1$ and use two first component of the rest of oscillatory neurons to represent normalised pixel values as scalar products $\boldsymbol{\mu}_{1}^{\top} \boldsymbol{\mu}_{j}$ (see Appendix~\ref{appendix:scalar_products_encoding}), for non-associative training we use $\boldsymbol{\Omega} = 0$ and rotate vectors for each batch of data randomly during training, for associative training we use trainable $\boldsymbol{\Omega}$, on the fine-tuning stage only output of the Hopfield network directly contributed to the loss function.

\begin{table}[t]
\caption{Number of parameters for Kuramoto subnetworks.}
\centering
\begin{tabular}{lcc}
\toprule
$k$ & model size \\
\midrule
$150$ & $\simeq 17\times 10^{4}$ \\
$250$ & $\simeq 28 \times 10^{4}$\\
$350$ & $\simeq 39 \times 10^{4}$\\
$450$ & $\simeq 50 \times 10^{4}$\\
\bottomrule
\end{tabular}
\label{table:N_params_Kuramoto}
\end{table}

\end{document}